\documentclass[acmsmall,10pt,screen]{acmart}

\setcopyright{rightsretained}
\acmPrice{}
\acmDOI{10.1145/3547648}
\acmYear{2022}
\copyrightyear{2022}
\acmSubmissionID{icfp22main-p85-p}
\acmJournal{PACMPL}
\acmVolume{6}
\acmNumber{ICFP}
\acmArticle{117}
\acmMonth{8}

\pdfoutput=1 

\usepackage{mathpartir}
\usepackage{proof}

\usepackage{amsmath}
\usepackage{enumitem}
\usepackage{thmtools}
\usepackage{thm-restate}
\usepackage{bm}
\usepackage{booktabs}
\usepackage{framed}
\usepackage[utf8]{inputenc}
\usepackage{mdframed}
\usepackage{graphicx}
\usepackage{multirow}
\usepackage{subcaption}
\usepackage{tabularx}
\usepackage[textsize=tiny]{todonotes}
\renewcommand{\todo}[1]{}

\usepackage{xspace}

\usepackage{cmap}
\usepackage[T1]{fontenc}

\usepackage{wrapfig}

\usepackage[scaled=0.75]{beramono}

\usepackage{nanoml}

\newcommand{\code}[1]{\mbox{\lstinline[language=nanoml,basicstyle=\footnotesize\ttfamily,columns=flexible,mathescape=true]^#1^}}

\newcommand{\mandate}{\textsc{Mandate}\xspace}

\newcommand{\eg}{\textit{e.g.:\ }}
\newcommand{\where}{\textbf{where}\xspace}

\newcommand{\secref}[1]{\S \ref{#1}}
\newcommand{\appref}[1]{App. \ref{#1}}
\newcommand{\figref}[1]{Fig. \ref{#1}}

\newcommand{\keyterm}[1]{\textbf{#1}}

\newtheorem{property}{Property}
\newtheorem{observation}{Observation}
\newtheorem{assumption}{Assumption}

\newcommand{\set}[1]{\ensuremath{\textsf{#1}}}
\newcommand{\p}{^\prime}
\newcommand{\pp}{^{\prime\prime}}
\newcommand{\ppp}{^{\prime\prime\prime}}

\newcommand{\gramdef}{\ensuremath{\mathrel{\mathord{:}\mathord{:=}}}}
\newcommand{\sor}{\ensuremath{\mathrel{|}}}
\newcommand{\gramalt}{\sor}

\newcommand{\IMP}{\textsc{IMP}\xspace}

\newcommand{\true}{\textbf{true}}
\newcommand{\false}{\textbf{false}}

\newcommand{\ite}[3]{\textbf{if }#1\textbf{ then }#2\textbf{ else }#3}
\newcommand{\while}[2]{\textbf{while }#1\textbf{ do }#2}

\newcommand{\assign}[2]{#1\text{ := }#2}
\newcommand{\impskip}{\textbf{skip}}

\newcommand{\seq}[2]{#1 ; #2}
\newcommand{\error}{\textbf{error}}

\newcommand{\lockstep}[2]{#1 \parallel #2}

\newcommand{\stmtvalname}{Stmt. Values}
\newcommand{\stmtvalvar}{w}

\newcommand{\mtval}{\textsf{Val}}
\newcommand{\mtnonval}{\textsf{NonVal}}
\newcommand{\mtall}{\textsf{All}}

\newcommand{\valstar}{\star_\mtval}
\newcommand{\nonvalstar}{\star_\mtnonval}
\newcommand{\allstar}{\star_\mtall}

\newcommand{\sosto}{\leadsto}
\newcommand{\build}[1]{#1}
\newcommand{\letstepto}[3]{\textbf{let}\ [#1 \sosto #2]\ \textbf{in}\ #3}
\newcommand{\letcomp}[3]{\textbf{let}\ #1 = #2\ \textbf{in}\ #3}

\newcommand{\with}{\mathop{\big |}}

\newcommand{\hole}{\square}
\newcommand{\plug}[2]{#1[#2]}

\newcommand{\halt}{\textbf{emp}}
\newcommand{\shortcxtfr}[1]{\left[ #1 \right]}
\newcommand{\cxtfr}[2]{\left[#1 \rightarrow #2\right]}
\newcommand{\frcomp}{\circ}

\newcommand{\pamto}{\mathrel{\hookrightarrow}}

\newcommand{\up}{\uparrow}
\newcommand{\down}{\downarrow}
\newcommand{\updown}{\updownarrow}

\newcommand{\pam}[3]{\,\am{#1}{#2}\!\scalebox{1.2}{\ensuremath{#3}}\,}

\newcommand{\am}[2]{\left< #1 \with #2 \right>}
\newcommand{\amto}{\mathrel{\rightarrow}}

\newcommand{\unfusedamto}{\longrightarrow}

\newcommand{\abstr}[1]{\widehat{#1}}

\newcommand{\partialfunc}{\rightharpoonup}

\newcommand{\emptyenv}{\emptyset}
\newcommand{\upd}[3]{#1[#2\rightarrow#3]}

\newcommand{\narrowsto}{\rightsquigarrow}
\newcommand{\absred}[1]{\mathrel{\underset{#1}{\abstr{\amto}}}}
\newcommand{\absnarrow}[1]{\mathrel{\underset{#1}{\abstr{\narrowsto}}}}

\setlist[itemize]{leftmargin=*}
\setlist[enumerate]{leftmargin=*}

\setcopyright{none}             

\begin{document}

\title{Automatically Deriving Control-Flow Graph Generators from Operational Semantics}


\author{James Koppel}
\affiliation{
  \institution{MIT}            
  \city{Cambridge}
  \state{MA}
  \country{USA}
}
\email{jkoppel@mit.edu}          

\author{Jackson Kearl}
\affiliation{
  \institution{MIT}
  \city{Cambridge}
  \state{MA}
  \country{USA}
}

\author{Armando Solar-Lezama}
\affiliation{
  \institution{MIT}            
  \city{Cambridge}
  \state{MA}
  \country{USA}
}
\email{asolar@csail.mit.edu}          


\begin{abstract}
We develop the first theory of control-flow graphs from first principles, and use it to create an algorithm for automatically synthesizing many variants of control-flow graph generators from a language's operational semantics. Our approach first introduces a new algorithm for converting a large class of small-step operational semantics to an abstract machine. It next uses a technique called ``abstract rewriting'' to automatically abstract the semantics of a language, which is used both to directly generate a CFG from a program (``interpreted mode'') and to generate standalone code, similar to a human-written CFG generator, for any program in a language. We show how the choice of two \textit{abstraction} and \textit{projection} parameters allow our approach to synthesize several families of CFG-generators useful for different kinds of tools. We prove the correspondence between the generated graphs and the original semantics. We provide and prove an algorithm for automatically proving the termination of interpreted-mode generators. In addition to our theoretical results, we have implemented this algorithm in a tool called \mandate, and show that it produces human-readable code on two medium-size languages with $60-80$ rules, featuring nearly all intraprocedural control constructs common in modern languages. We then show these CFG-generators were sufficient to build two static analyses atop them. Our work is a promising step towards the grand vision of being able to synthesize all desired tools from the semantics of a programming language.
\end{abstract}

\ccsdesc[500]{Theory of computation~Abstract machines}
\ccsdesc[500]{Theory of computation~Operational semantics}

\keywords{}  

\maketitle

\section{Introduction}
\label{sec:introduction}

%
%


Many programming tools use control-flow graphs, from compilers to model-checkers. They provide a simple way to order the subterms of a program, and are usually taken for granted. According to folklore, their definition is simple: ``control-flow graphs are an abstraction of control-flow."

\setlength{\columnsep}{8pt}
\begin{wrapfigure}{r}{0.21\textwidth}
\vspace{-1em}
\begin{mdframed}
\begin{center}
\vspace{-0.4em}
\begin{lstlisting}[language=Algol,basicstyle=\scriptsize\sffamily]
b := prec < 5;
if (b) then
  print("(")
else
  skip;
print(left);
print("+");
print(right);
if (b) then
  print(")")
else
  skip
\end{lstlisting}
\end{center}
\vspace{-0.5em}
\end{mdframed}
\vspace{-0.5em}
\caption{}
\label{fig:example-cfg-code}
\vspace{-1.5em}
\end{wrapfigure}

In fact, as we shall argue, CFGs are not well understood, and their definition is not so simple. Consider: Even for a single language, no two tools generate the same CFG for the same program, and we have found no prior attempt to define what it means for a given CFG to correctly abstract a program. Before diving deeper into the need for a theory of CFGs, let us illustrate the nuances of CFG-generation: \figref{fig:example-cfg-code} is a fragment of a pretty-printer. How would a CFG for it look? Here are three possible answers for different kinds of tools:

\begin{enumerate}
\item Compilers want small graphs that use little memory. A compiler may only give one node per basic block, giving the graph in \figref{fig:cfg-variants-bb}.

\item Many static analyzers give abstract values to the inputs and result of every expression. To that end, frameworks such as Polyglot \cite{nystrom2003polyglot} and IncA \cite{szabo2016inca} give two nodes for every expression: one each for entry and exit. \figref{fig:cfg-variants-expr} shows part of this graph.

\item Consider an analyzer that proves this program's output has balanced parentheses. It must show that there are no paths in which one if-branch is taken but not the other. This can be easily written using a path-sensitive CFG that partitions on the value of \texttt{b} (\figref{fig:cfg-variants-path-sensitive}). Indeed, in \secref{sec:mandate}, we build such an analyzer atop path-sensitive CFGs in under $50$ lines of code.
\end{enumerate}

\noindent All three tools require separate CFG-generators.

\todo{Why is the centering off?}
\begin{figure}
\vspace{-0.6em}
\centering
\begin{subfigure}[t]{0.25\textwidth}
\centering
\includegraphics[scale=0.29]{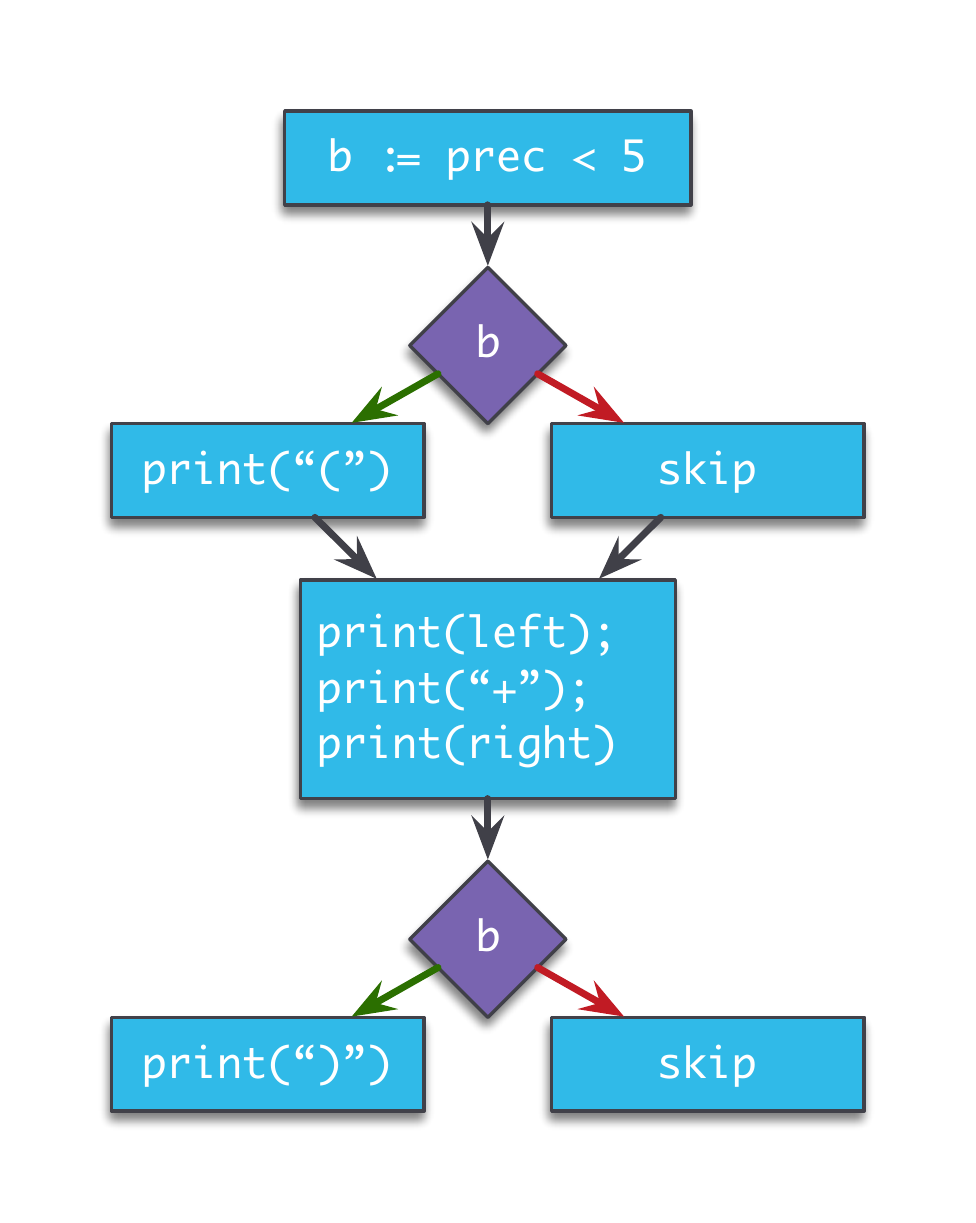}
\vspace{-0.75em}
\subcaption{}
\label{fig:cfg-variants-bb}
\end{subfigure}
~ 
\begin{subfigure}[t]{0.33\textwidth}
\centering
\includegraphics[scale=0.33]{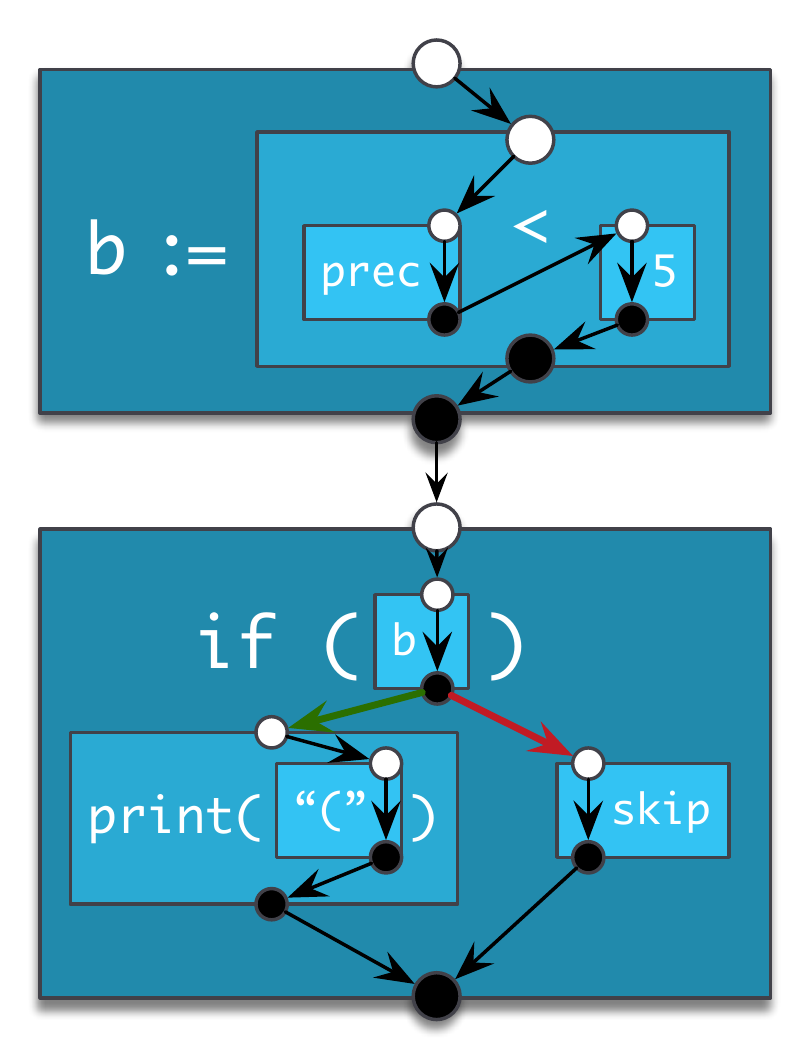}
\vspace{-0.75em}
\subcaption{}
\label{fig:cfg-variants-expr}
\end{subfigure}
~
\begin{subfigure}[t]{0.31\textwidth}
\centering
\includegraphics[scale=0.29]{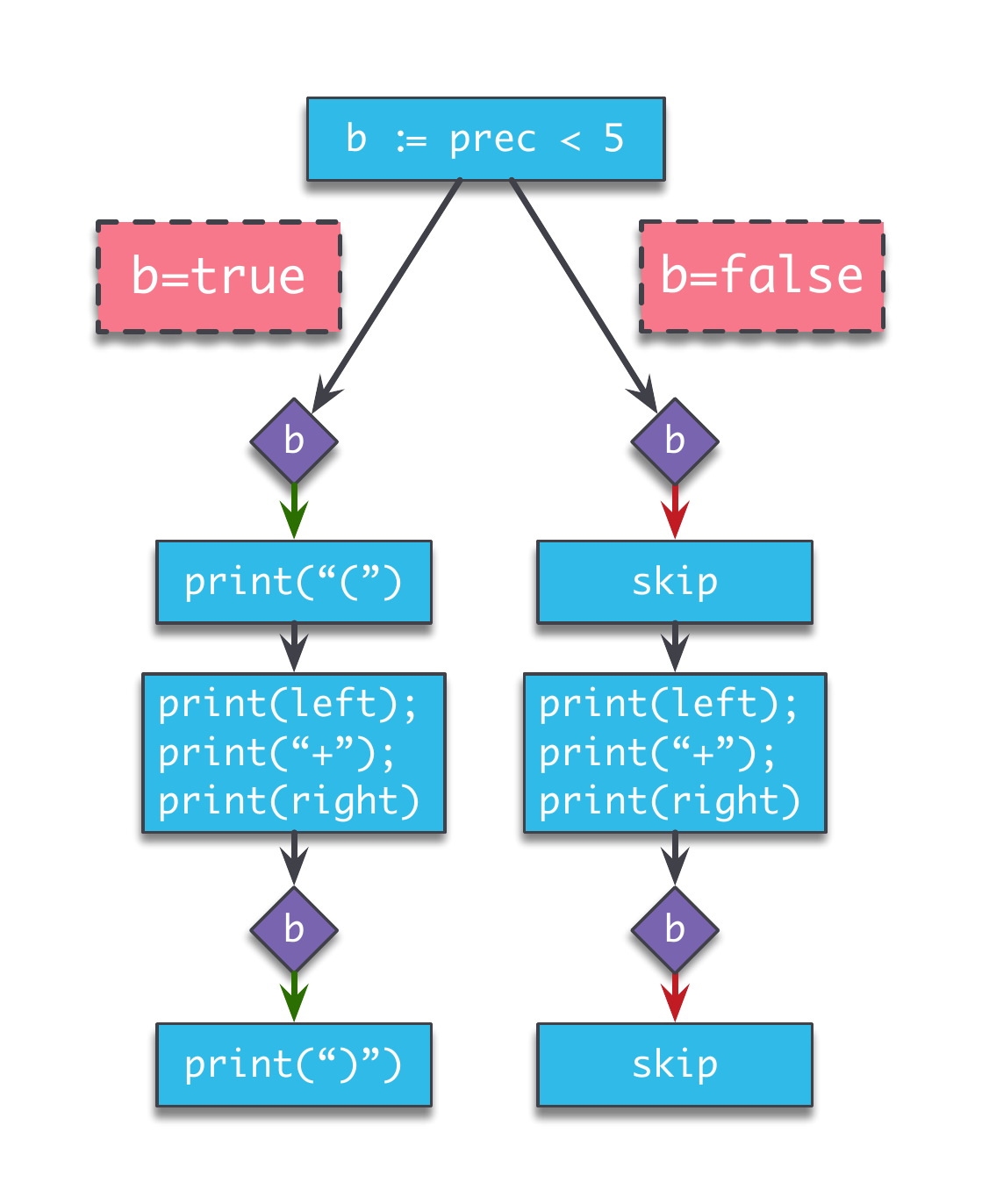}
\vspace{-0.75em}
\subcaption{}
\label{fig:cfg-variants-path-sensitive}
\end{subfigure}

\vspace{-0.75em}
\caption{Variants of control-flow graphs. Colors are for readability.}
\vspace{-1em}
\label{fig:cfg-variants}
\end{figure}

\paragraph{Whence CFGs?}

\todo{Martin doesn't like this paragraph. What to do? I think we're keeping it; no complaints yet.}

What if we had a formal semantics (and grammar) for a programming language? In principle, we should be able to automatically derive all tools for the language. In this dream, only one group needs to build a semantics for each, and then all tools will automatically become available --- and semantics have already been built for several major languages \cite{bogdanas2015k,park2015kjs,DBLP:conf/pldi/HathhornER15}. In this paper, we take a step towards that vision by developing the formal correspondence between semantics and control-flow graphs, and use it to automatically derive CFG generators from a large class of operational semantics.

While operational semantics define each step of computation of a program, the correspondence with control-flow graphs is not obvious. The ``small-step'' variant of operational semantics defines individual steps of program execution. Intuitively, these steps should correspond to the edges of a control-flow graph. In fact, we shall see that many control-flow edges correspond to the second half of one rule, and the first half of another. We shall similarly find the nodes of a control-flow graph correspond to neither the subterms of a program nor its intermediate values. Existing CFG generators skip these questions, taking some notion of labels or ``program points'' as a given (e.g.: \cite{shivers1991control}). We instead develop CFGs from first principles, and, after much theory, discover that \textbf{``a CFG is a projection of the transition graph of abstracted abstract machine states.''}

\paragraph{Abstraction and Projection and Machines}

The first insight is to transform the operational semantics into another style of semantics, the abstract machine, via a new algorithm. Evaluating a program under these semantics generates an infinite-state transition system with recognizable control-flow. Typically at this point, an analysis-designer would manually specify some kind of abstract semantics which collapses this system into a finite-state one. Our approach does this automatically by interpreting the concrete semantics differently, using an obscure technique called \textit{abstract rewriting}. From this reduced transition system, a familiar structure emerges: we have obtained our control-flow graphs!


Now all three variants of control-flow graph given in this section follow from different abstractions of the abstract machine, followed by an optional {\it projection}, or merging of nodes. With this approach, we can finally give a formal, proven correspondence between the operational semantics and all such variants of a control-flow graph.

\paragraph{\mandate: A CFG-Generator Generator}

The primary goal of this paper is to develop the first theory of CFGs from first principles. Yet our theory immediately lends itself to automation. We have implemented our approach in \mandate, the first control-flow-graph generator generator. \mandate takes as input the operational semantics for a language, expressed in an embedded Haskell DSL that can express a large class of systems, along with a choice of abstraction and projection functions. It then outputs a control-flow-graph generator for that language. By varying the abstraction and projection functions, the user can generate any of a large number of CFG-generators.

\mandate has two modes. In the \textit{interpreted mode}, \mandate abstractly executes a program with its language's semantics to produce a CFG. For cases where the control-flow of a node is independent from its context (e.g.: including variants (1) and (2) but not (3)), \mandate's \textit{compiled mode} can output a CFG-generator as a short program, similar to what a human would write (\figref{fig:front-page-result}).

\begin{figure*}
\begin{center}
\includegraphics[scale=0.65]{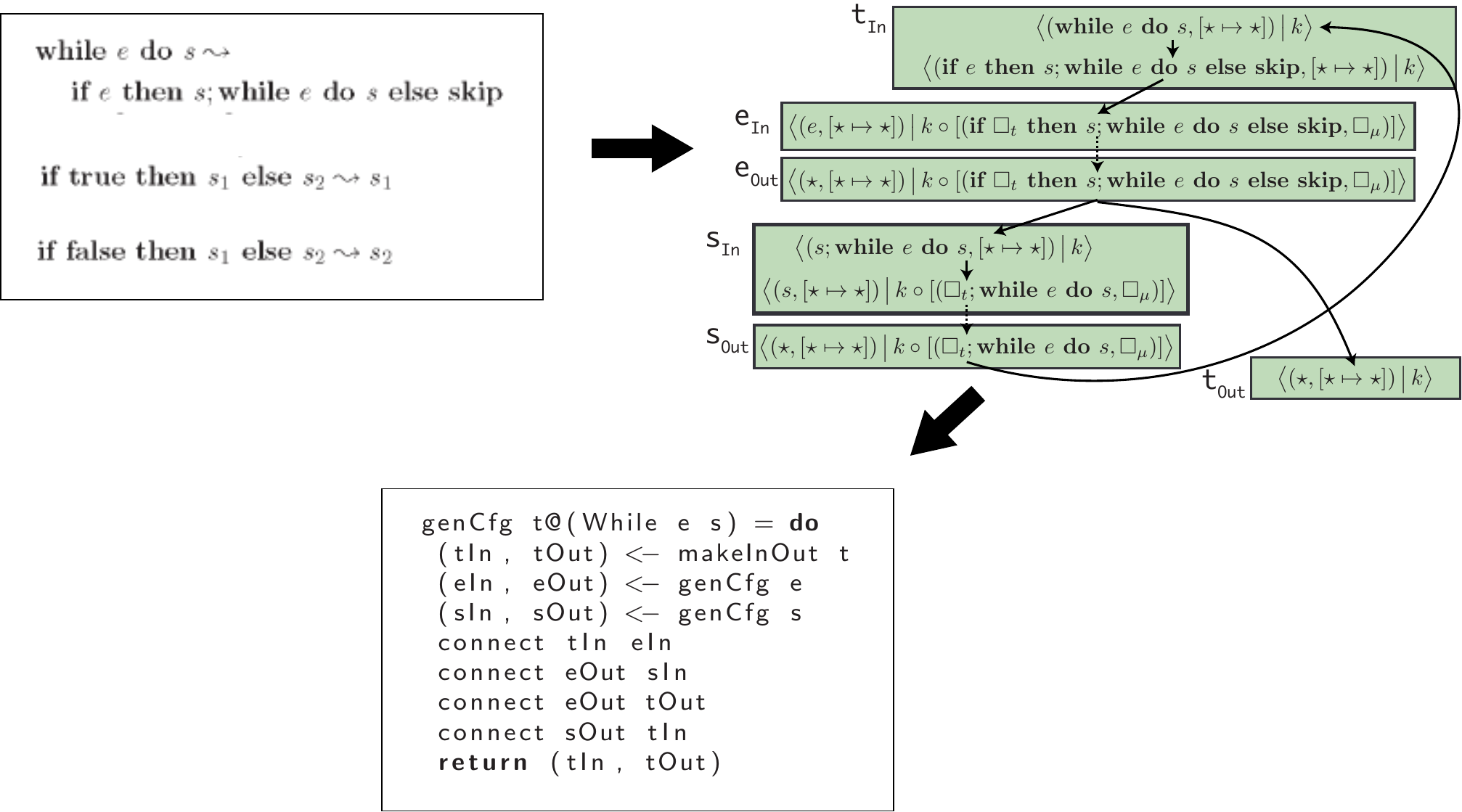}
\end{center}
\vspace{-1em}
\caption{(Top left) SOS rules for loops and conditionals (Top right) A graph-pattern generated from these rules, describing the control-flow of all while-loops (Bottom) Generated CFG-generation code}
\label{fig:front-page-result}
\end{figure*}

We have evaluated \mandate on several small languages, as well as two larger ones. The first is \textsc{Tiger}, an Algol-family language used in many compilers courses, and made famous by a textbook of \citet{DBLP:books/cu/Appel1998ml}. The second is \textsc{MITScript} \cite{mitscript}, a JavaScript-like language with objects and higher-order functions used in an undergraduate JIT-compilers course. While these are pedagogical languages without the edge-cases of C or SML, they nonetheless contain all common control-flow features except exceptions. And, since previous work on conversion of semantics features small lambda calculi, ours are the largest examples of automatically converting a semantics into a different form, by a large margin.

In summary, our work makes the following contributions:

\begin{enumerate}
\item A formal and proven correspondence between the operational semantics and many common variations of CFGs, giving the first from-first-principles theoretical explanation of CFGs.
\item An elegant new algorithm for converting small-step structural operational semantics into abstract machines.
\item An algorithm which derives many types of control-flow graph generators from an abstract machine, determined by choice of abstraction and projection functions, including standalone generators which execute without reference to the semantics (``compiled mode'').
\item An ``automated termination proof,'' showing that, if the compiled-mode CFG-generator terminates (run once per language/abstraction pair), then so does the corresponding interpreted-mode CFG generator (run once per program).
\item The first CFG-generator generator, \mandate, able to automatically derive many types of CFG generators for a language from an operational semantics for that language, and successfully used to generate CFG-generators for two rich languages. The generated CFG-generators were then used to build two static-analyzers.
\end{enumerate}

\noindent Further, our approach using \textit{abstract rewriting} offers great promise in deriving other artifacts from language semantics.

\section{Control-Flow Graphs for \IMP}
\label{sec:overview}

\begin{wrapfigure}{l}{0.58\textwidth}
\begin{center}
\vspace{-1.5em}
\includegraphics[scale=0.45]{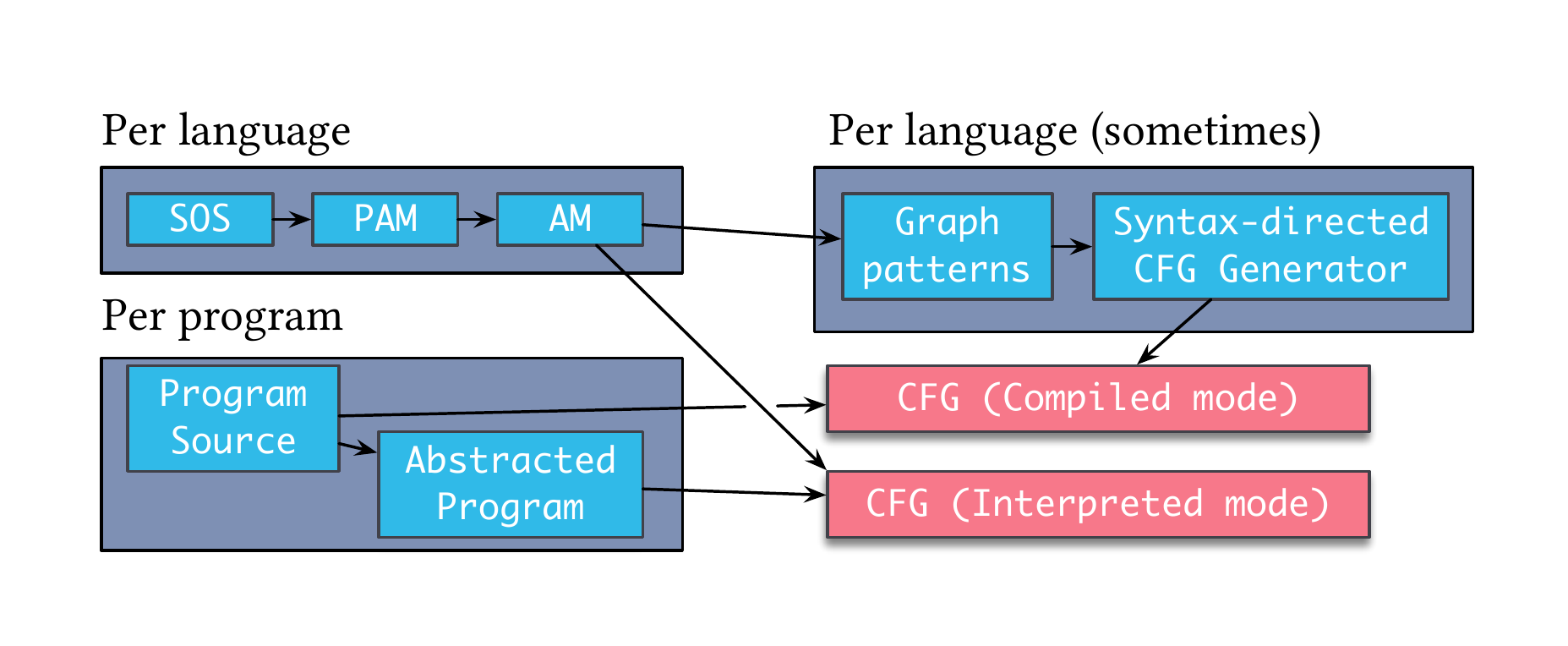}
\end{center}
\vspace{-2em}
\caption{Dataflow of our approach}
\label{fig:architecture}
\vspace{-0.6em}
\end{wrapfigure}

We shall walk through a simple example of generating a control-flow graph generator, using a simple imperative language called \IMP. \IMP features conditionals, loops, mutable state, and two binary operators. \figref{fig:imp-syntax} gives the syntax. In this syntax, we explicitly split terms into values and non-values to make the rules easier to write. We will do this more systematically in \secref{sec:sos}.

\begin{figure}
\vspace{-0.5em}
\begin{mdframed}
\begin{small}
\begin{center}
\vspace{-0.5em}
\begin{mathpar}
\hspace {-0.4cm}\begin{tabular}{rlcl}

\text{Variables} & \multicolumn{3}{l}{$x,y,\dots$ $\in \set{Var}$} \\
\text{Expr. Values} & $v$ & \gramdef & $n \in \set{Int} \sor \true \sor \false$ \\
\text{Expressions} & $e$ & \gramdef & $v \sor x \sor e + e \sor e < e$ \\
\text{\stmtvalname} & $\stmtvalvar$ & \gramdef & \impskip \\
\text{Statements} & $s$ & \gramdef & $\stmtvalvar \sor \seq{s}{s} \sor \while{e}{s}$ \\
 & & &  $\sor \assign{x}{e} \sor \ite{e}{s}{s} $

\end{tabular}
\end{mathpar}
\vspace{-0.5em}
\end{center}
\end{small}
\end{mdframed}
\vspace{-0.5em}
\caption{Syntax of \IMP}
\vspace{-0.7em}
\label{fig:imp-syntax}
\end{figure}

The approach proceeds in three phases, corresponding roughly to the large boxes in \figref{fig:architecture}. In the first phase (top-left box / \secref{sec:overview-sos-to-am}), we transform the semantics of \IMP into a form that reveals the control flow. In the second phase (bottom-left box / \secref{sec:overview-cfg-gen}), we show how to interpret these semantics with abstract rewriting \cite{bert1993abstract} to obtain CFGs. In the finale (top-right box / \secref{sec:overview-graph-pat}), we show how to use abstract rewriting to discover facts about all \IMP programs, resulting in human-readable code for a CFG-generator.

\subsection{Getting Control of the Semantics}
\label{sec:overview-sos-to-am}

\paragraph{Semantics} The language has standard semantics, so we only show a few of the rules, given as structural operational semantics (SOS). Each rule relates an old term and environment $(t, \mu)$ to a new one $(t\p, \mu\p)$ following one step of execution. As our first \textbf{running example}, we use the rules for assignments. We will later introduce our two other running examples: addition and while-loops. Together, these cover all features of semantic rules: environments, external semantic functions, branching, and back-edges.

\label{sec:contains-assn-rules}
\begin{center}
\small{
\begin{tabular}{c}

%
  \infer[AssnCong]
  {(\assign{x}{e},\mu) \sosto (\assign{x}{e\p}, \mu\p)}
 {(e, \mu) \sosto (e\p, \mu\p)}
  \\[1.0em]
    \infer[AssnEval]
    {(\assign{x}{v}, \mu) \sosto (\impskip, \upd{\mu}{x}{v})}
    {}
  
\end{tabular}
}
\end{center}


\noindent These rules give the basic mechanism to evaluate a program, but don't have the form of stepping from one subterm to another, as a control-flow graph would. So, from these rules, our algorithm will automatically generate an \textit{abstract machine} (AM). This machine acts on states $\am{(t,\mu)}{K}$. $K$ is the \textit{context} or \textit{continuation}, and represents what the rest of the program will do with the result of evaluating $t$. $K$ is composed of a stack of \textit{frames}. While the general notion of frames is slightly more complicated (\secref{sec:pam}), in most cases, a frame can be written as e.g.: $\shortcxtfr{(\assign{x}{\hole_t}, \hole_\mu)}$, which is a frame indicating that, once the preceding computation has produced a term and environment $(t\p, \mu\p)$, the next step will be to evaluate $(\assign{x}{t\p}, \mu\p)$. Our algorithm generates the following rules. These match the textbook treatment of AMs \cite{felleisen2009semantics}.

\vspace{-0.5em}
\label{sec:contains-imp-machine}
\begin{small}
\begin{mathpar}

  \begin{array}{l@{\ \amto\ }l}
  
  
    
    \am{(\assign{x}{e}, \mu)}{k} & \am{(e, \mu)}{k \frcomp \shortcxtfr{(\assign{x}{\hole_t}, \hole_\mu)}} \\
    \am{(v, \mu)}{k\frcomp \shortcxtfr{(\assign{x}{\hole_t}, \hole_\mu)}}
    & \am{(\impskip, \upd{\mu}{x}{v})}{k} \\
    
  \end{array}
\end{mathpar}
\end{small}

\noindent The first AM rule, on seeing an assignment $\assign{x}{e}$, will focus on $e$. Other rules, not shown, then reduce $e$ to a value. The second then takes this value and evaluates the assignment.

While these rules have been previously hand-created, this work is the first to derive them automatically from SOS. These two SOS rules become two AM rules. A naive interpretation is that the first SOS rule corresponds to the first AM rule, and likewise for the second. After all, the left-hand sides of the first rules match, as do the right-hand sides of the second. But notice how the RHSs of the firsts do not match, nor do the LHSs of the seconds. The actual story is more complicated. This diagram gives the real correspondence:

\begin{center}
\vspace{0.2em}
\includegraphics[scale=0.8]{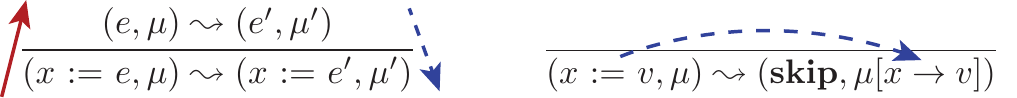}
\vspace{0.2em}
\end{center}

The first AM rule is simple enough: it corresponds to the solid arrow, the first half of the first SOS rule. But the second AM rule actually corresponds to the two dashed arrows, \textsc{AssnEval} and the second half of \textsc{AssnCong}. We shall find that jumping straight from SOS to AM skips an intermediate step, which treats each of the three arrows separately.

This fusing of \textsc{AssnEval} and \textsc{AssnCong} happens because only two actions can follow the second half of \textsc{AssnCong}: \textsc{AssnEval}, and the first half of \textsc{AssnCong} --- which is the inverse of the second half. Hence, only the pairing of \textsc{AssnCong} with \textsc{AssnEval} is relevant, and two AM rules are enough to describe all computations. And the second AM rule actually does two steps in one: it plugs $(v,\mu)$ into $\shortcxtfr{(\assign{x}{\hole_t}, \hole_\mu)}$ to obtain $(\assign{x}{v}, \mu)$, and then evaluates this result. So, by fusing these rules, the standard presentation of AM obscures the multiple underlying steps, hiding the correspondence with the SOS.

This insight --- that SOS rules can be split into several parts --- powers our algorithm to construct abstract machines. The algorithm first creates a representation called the \textit{phased abstract machine}, or PAM, which partitions the two SOS rules into three parts, corresponding to the diagram's three arrows, and gives each part its own rule. The algorithm then fuses some of these rules, creating the final AM, and completing the first stage of our CFG-creation pipeline.

\begin{wrapfigure}{R}{0.42\textwidth}
\vspace{-0.7em}
\begin{mdframed}
\vspace{-0.2em}
\begin{center}
\hspace{-042cm}\begin{lstlisting}[language=Haskell,basicstyle=\footnotesize\sffamily]
valueIrrelevance t =
  mapTerm valToStar t
 where
  valToStar (Val _ _) = ValStar
  valToStar t         = t
\end{lstlisting}
\end{center}
\vspace{-0.2em}
\end{mdframed}
\vspace{-0.8em}
\caption{}
\vspace{-1em}
\label{fig:val-irrelevance-code}
\end{wrapfigure}

\secref{sec:pam} explains PAM in full, while \secref{sec:sos-to-pam} and \secref{sec:pam-to-am} give the algorithm for creating abstract machines. \secref{sec:correctness-body} shows correctness.

\subsection{Run Abstract Program, Get CFG}
\label{sec:overview-cfg-gen}

The AM rules show how focus jumps into and out of an
assignment during evaluation --- the seeds of control-flow.  But these
transitions are not control-flow edges; there are still a few
important differences. The AM allows for infinitely-many
possible states, while control-flow graphs have finite numbers of
states, potentially with loops. The AM
executes deterministically, with every state stepping into one other
state. Even though we assume determistic languages, even for those,
the control-flow graph may branch. We will see how abstraction solves
both of these issues, turning the AM into the interpreted-mode control-flow graph generator.

To give a complete example, we'll also need to evaluate an expression. Here is the rule for variable lookups, which looks up $y$ in the present environment:

\vspace{-1em}
\begin{small}
\begin{center}
\begin{mathpar}
    \am{(y, \mu)}{k}\ \amto\  \am{(\mu(y), \mu)}{k} \\
\end{mathpar}
\end{center}
\end{small}
\vspace{-1em}

Now consider the statement $\assign{x}{y}$. It can be executed with an infinite number of environments: the starting configuration $(\assign{x}{y}, [y\mapsto 1])$ results in $(\impskip, [y\mapsto 1, x\mapsto 1])$; $(\assign{x}{y}, [y\mapsto 2])$ yields $(\impskip, [y\mapsto 2, x\mapsto 2])$; etc.

To yield a control-flow graph, we must find a way to compress this infinitude of possible states into a finite number. The value-irrelevance abstraction, given by the code in \figref{fig:val-irrelevance-code}, replaces all values with a single \textit{abstract value} representing any of them: $\star$. Under this abstraction, all starting environments for this program will be abstracted into the single environment $[y\mapsto\star]$.

In a typical use of abstract interpretation, at this point a new abstract semantics must be manually defined in order to define executions on the abstract state. However, with this kind of syntactic abstraction, our system can interpret the exact same AM rules on this abstract state, a process called \textit{abstract rewriting}. Now, running in fixed context $K$, there is only one execution of this statement.

\vspace{0.6em}
\begin{small}
\qquad \begin{tabular}{crcl}
      &  $\am{(\assign{x}{y},[y\mapsto\star])}{K}$ & $\amto$ &  $\am{(y, [y\mapsto\star])}{K \frcomp\shortcxtfr{(\assign{x}{\hole_t}, \hole_\mu)}}$ \\[0.5em]
 $\amto$  &  $\am{(\star, [y\mapsto\star])}{K \frcomp\shortcxtfr{(\assign{x}{\hole_t}, \hole_\mu)}}$  & $\amto$ &   $\am{(\impskip, [y\mapsto\star, x\mapsto\star])}{K }$
 \end{tabular}
 \end{small}
\vspace{0.6em}

\noindent This abstract execution is divorced from any runtime values, yet it still shows the flow of control entering and exiting the assignment and expression--- exactly as in the expression-level control-flow graph from \figref{fig:cfg-variants-expr}. And thus we can take these abstract states to be our control-flow nodes. The CFG is an abstraction of the transitions of the abstract machine.

Note that, because there are only finitely-many abstract states, this also explains loops in control-flow graphs: looping constructs lead to repeated states, which leads to back-edges in the transition graph. And abstractions also account for branching. Consider the rules for conditionals:

\vspace{-1em}
\begin{small}
\begin{eqnarray*}
\am{(\true, \mu)}{k\frcomp \shortcxtfr{(\ite{\hole_t}{s_1}{s_2},\hole_\mu)}}\,\amto\,\am{(s_1, \mu)}{k} \\
\am{(\false, \mu)}{k\frcomp \shortcxtfr{(\ite{\hole_t}{s_1}{s_2},\hole_\mu)}}\,\amto\,\am{(s_2, \mu)}{k} \\
\end{eqnarray*}
\end{small}
\vspace{-2em}

\noindent Under the value-irrelevance abstraction, the condition of an if-statement will evaluate to a configuration of the form $(\star, \mu)$. Because $\star$ can represent both $\true$ and $\false$, \textbf{both} of the above rules will match. This gives control-flow edges from the if-condition into both branches, exactly as desired.

The value-irrelevance abstraction gives an expression-level CFG. But it is not the only choice of abstraction. \secref{sec:cfg-abstraction} presents the CFG-derivation algorithm, and shows how other choices of abstraction lead to other familiar control-flow graphs, and also introduces projections, which give the CFG-designer the ability to specify which transitions are ``internal'' and should not appear in the CFG.

\subsection{A Syntax-Directed CFG-Generator}
\label{sec:overview-graph-pat}

The previous section showed how to abstract away the inputs and
concrete values of an execution, turning a program into a CFG. The
per-program state-exploration of this algorithm is
necessary for a path-sensitive CFG-generator, which may create an
arbitrary number of abstract states from a single AST node. But for abstractions which
discard all contextual information, only a few additional small ingredients
are needed to generate a single artifact that describes the control-flow of all instances of a
given node-type. This is done once per language, yielding the
compiled-mode CFG-generator. We demonstrate how this works for while-loops, showing how our
algorithm can combine many rules to infer control-flow, abstracts away
extra steps caused by the internal details of the semantics, and can
discover loops in the control-flow even though they are not explicit
in the rules.

The semantics of while-loops in \IMP are given in terms of other language constructs,
by the rule $\while{e}{s} \sosto
\ite{e}{(\seq{s}{\while{e}{s}})}{\impskip}$. Consider an AM state evaluating an arbitrary while-loop $\while{e}{s}$
in an arbitrary context $k$, with an arbitrary abstract environment $\mu$. Such a
state can be written $\am{(\while{e}{s}, \mu)}{k}$. Any such
$\mu$ can be represented by the ``top'' environment
$[\star\mapsto\star]$. This means that all possible transitions from
any while-loop can be found by finding all rules that could match
anything of the form $\am{(\while{e}{s},
  [\star\mapsto\star])}{k}$. Repeatedly expanding these transitions
results in a \textit{graph pattern} describing the control-flow for
every possible while-loop.

However, merely searching for matching rules will not result in a
finite graph, because of states like  $\am{(e,
  [\star\mapsto\star])}{k}$. These states, which represents the intent
to evaluate the unknown subterm $e$, can match rules for any
expression. Instead, we note that, for any given $e$, other rules
(corresponding to other graph patterns) would evaluate its
control-flow, and their results can all be soundly overapproximated by
assuming $e$ is eventually reduced to a value. Hence, when the search
procedure encounters such a state, it instead adds a ``transitive
edge'' to a state $\am{(\star, [\star\mapsto\star])}{k}$. With this
modification, the search procedure finds only $8$ unique states for
while-loops. It hence terminates in the graph pattern of \figref{fig:front-page-result} (with dotted lines for transitive edges), which describes the control-flow of all possible while-loops. From this pattern, one could directly generate a CFG fragment for any given while-loop by instantiating $e$, $s$, and $k$. In combination with the graph patterns for all other nodes, this yields a control-flow graph with a proven correspondence to the original program.

But, from these graph patterns, it is also straightforward to output
code for a syntax-directed CFG-generator similar to what a human would
write. Our code-generator traverses this graph pattern,
identifying some states as the entrance and exit nodes of the entire
while-loop and its subterms. All other states are considered internal
steps which get merged with the enter and exit states (via a
\textit{projection}), resulting in a few ``composite''
states. \figref{fig:front-page-result} shows how the code-generator
groups and labels the states of this graph pattern as well as the generated code\footnote{This is verbatim \mandate output except that, in actual output, (1) the \code{connect} statements are in no particular order, and (2) the actual return value is \code{([tIn], [tOut])}, as, in general, AST nodes such as conditionals may have multiple final CFG nodes.}.

After many steps transforming and analyzing the semantics of \IMP, the
algorithm has finally boiled down all aspects of the control flow into
concise, human-readable code --- for an expression-level
CFG-generator. To generate a statement-level CFG-generator, the user
must merely re-run the algorithm again with a different
abstraction. For while-loops, the resulting pattern and code are
similar to those of \figref{fig:front-page-result}, except that they skip the evaluation
of $e$.

The last few paragraphs already gave most of the details of graph-pattern construction. \secref{sec:syntax-directed}
gives the remaining details, while
\appref{app:graph-pattern-correctness} proves the algorithm's correctness.

\section{From Operational Semantics to Abstract Machines}
\label{sec:sos-to-am}

The first step in our algorithm is to convert the structural operational semantics for a language into an abstract machine which has a clear notion of control-flow. Surprisingly, no prior algorithm for this exists (see discussion in \secref{sec:really-explains-danvy} and \appref{app:related-work}). This section hence presents \textbf{the first algorithm to convert SOS to AM}, which works on any language satisfying the conditions in \secref{sec:assumptions}. Our algorithm is unique in its use of a
new style of semantics as an intermediate form, the phased abstract
machine, which simulates a recursive program running the SOS rules. We
believe this formulation is particularly elegant and leads to simple
proofs, while being able to scale to the realistic languages
\textsc{Tiger} and \textsc{MITScript}.

There is a lot of notation required first. \secref{sec:terms-languages} gives a notation for all programming languages, while  \secref{sec:sos} gives an alternate notation for structural operational semantics, one more amenable to inductive transformation. \secref{sec:pam} and \secref{sec:am} describe the phased and orthodox abstract machines, while \secref{sec:sos-to-pam} and \secref{sec:pam-to-am} give the conversion algorithm. Correctness results are provided in \secref{sec:correctness-body} and \appref{app:correctness}.

\subsection{Setting and Assumptions}
\label{sec:assumptions}

We take as our starting point a transition relation on term/state pairs $(t, \mu) \sosto (t', \mu')$, defined using rules in the variant of structural operational semantics (SOS) to be given in \secref{sec:sos}. From this, our algorithms shall construct the Phased Abstract Machine, the Abstract Machine, and finally a CFG-generator. The middle step of PAM-to-AM conversion does the work of discovering which rules may follow which other rules, and imposes some additional requirements on the language semantics.

\begin{assumption}[Sanity of Values] All terms can be classified into either \textit{values} or \textit{nonvalues}, based only on the root node. For a rule, $(t,\mu) \sosto \text{rhs}$, the pattern $t$ must not match a value. If $t$ is a nonvalue, there are $\mu,t\p,\mu\p$ such that $(t,\mu)\sosto (t\p,\mu\p)$.
\end{assumption}

\begin{assumption}[Determinism]
For any $t, \mu$, there is at most one $t\p, \mu\p$ such that $(t,\mu) \sosto (t\p, \mu\p)$.
\end{assumption}

\begin{assumption}[No Up-Down Rules and All Up-Rules Invertible]
Discussed in \secref{sec:pam-to-am}.
\end{assumption}

The Sanity of Values assumption is useful in correctness proofs, but imposes no real burden; a language definition can generally be tweaked to satisfy it.

The Determinism assumption may seem strong, but it is in fact necessary because \textbf{nondeterministic languages lack a clear notion of control-flow}. Nondeterminism means that consecutive steps may occur in distant parts of a term: evaluating $((a * b) + (c - d))+(e / f)$ may e.g.: evaluate first the multiplication on the left, then the division on the right, and then the subtraction on the left. The control flow of such programs cannot be given as a transition system (i.e.: a graph), which assumes a single program point. Instead, it is properly given as a Petri Net, called the ``sea of nodes'' representation \cite{click1995combining}. We leave generating a sea-of-nodes-generator to future work.

The last assumption is most easily defined on the PAM, and discussed in the corresponding section. Its first part, ``No Up-Down Rules,'' translates easily to SOS: it combines the superficial constraint that any rule which recursively evaluates a subterm must do any additional processing before the recursion, and the stronger constraint that there can be no single step of a term which simultaneously steps two independent subterms. The ``All Up-Rules Invertible'' assumption is less easy to translate, but it effectively means that successive steps in the SOS will focus on the same subterm, which is key to identifying the present ``program point.''

\subsection{Terms and Languages}

\label{sec:terms-languages}

Our presentation requires a uniform notation for terms in all languages. \figref{fig:term-gram} gives this notation, describing both concrete terms as well as patterns used in rewrite rules. Throughout this paper, we use the notation $\overline{\cdot}$ to represent lists, so that, e.g.: $\overline{\text{term}}$ represents the set of lists of terms.

\begin{wrapfigure}{L}{0.51\textwidth}
\begin{mdframed}
\begin{center}
\begin{small}
\vspace{-0.2em}
\hspace{-0.4cm}\begin{tabularx}{\textwidth}{rcX}

$\text{const}$    & \gramdef & $n \in \set{Int} \sor \text{str} \in \set{String} $ \\
\text{sym} & \gramdef & $+, <, \textbf{if},\dots$ $\in \set{Symbol}$ \\
\text{mt} & \gramdef & $\mtval \sor \mtnonval \sor \mtall$ \hfill{(Match Types)} \\
& & $a,b,c,\dots$ $\in \set{Var}$ \hfill{(Raw Vars)} \\
$x$ & \gramdef & $a_{\text{mt}}$ \hfill{(Pattern Vars)}\\
$\text{term}$   & \gramdef & $ \text{nonval}(\text{sym}, \overline{\text{term}}) $ \\
                & \gramalt & $ \text{val}(\text{sym}, \overline{\text{term}}) $ \\
                & \gramalt & $ \text{const} $ \\
                & \gramalt & $ x $ \\
$\mu$             & \gramdef & $\set{State}_l$  \hfill{(Reduction State)} \\
$c$             & \gramdef & $(\text{term}, \mu)$ \hfill{(Configurations)}

\end{tabularx}
\vspace{-0.5em}
\end{small}
\end{center}
\end{mdframed}
\caption{Universe of terms.}
\label{fig:term-gram}
\vspace{-1em}
\end{wrapfigure}

A typical presentation of operational semantics will have variable names like $v_1$ and $e\p$, where $v_1$ implicitly represents a \textit{value} which cannot be reduced, while $e\p$ represents a \textit{nonvalue} which can. We formalize this distinction by marking each node either value or nonvalue, and giving each variable a \textit{match type} controlling whether it may match only values, only nonvalues, or either.



Each variable is specialized with one of three \keyterm{match types}. Variables of the form $a_\mtval$, $a_\mtnonval$, and $a_\mtall$ are called \keyterm{value}, \keyterm{nonvalue}, and \keyterm{general} variables respectively. Value variables may only match $\text{val}$ constructors and constants; nonvalue variables match only $\text{nonval}$ constructors; general variables match any. 


We use a shorthand to mimic typical presentations of semantics. We will use $e$ to mean a nonvalue variable, $v$ or $n$ to mean a value variable, and $t$ or $x$ for a general variable. But variables in a right-hand side will always be general variables unless said otherwise. For instance, in $e \sosto e\p$, $e$ is a nonvalue variable, while $e\p$ is a general variable.

Each internal node is tagged either \textit{val} or \textit{nonval}. For example, $1+1$ is shorthand for the term $\text{nonval}(+,\text{val}(\text{IntLit},1),\text{val}(\text{IntLit},1))$. The \IMP statement $\assign{x}{1}$ may ambiguously refer to either the concrete term $\text{nonval}(:=,``x",\text{val}(\text{IntLit},1))$ or to the pattern $\text{nonval}(:=,x_\mtall,\text{val}(\text{IntLit},1))$, but should be clear from context. Others' presentations commonly have a similar ambiguity, using the same notation for patterns and terms.

Each language $l$ is associated with a \keyterm{reduction state} $\set{State}_l$ containing all extra information used to evaluate the term. For example, $\set{State}_\IMP$ is the set of \textit{environments}, mapping variables to values. Formally:

\vspace{-1.7em}
\begin{small}
\begin{mathpar}
\begin{tabular}{lcl}
$\text{env}$ & \gramdef & $\emptyset \sor \upd{\text{env}}{\text{str}}{v}$ \\
\end{tabular}
\end{mathpar}
\end{small}
\vspace{-1.3em}

\noindent Environments are matched using associative-commutative-idempotent patterns \cite{baader1999term}, so that \eg the pattern $\upd{x}{``y"}{v}$ matches the environment $\upd{\upd{\emptyset}{``y"}{1}}{``z"}{2}$. The latter environment is abbreviated $[z\mapsto 2, y\mapsto 1]$. The environment-extension operator itself is not commutative. Specifically, it is right biased, e.g.: $\upd{\upd{\emptyset}{``z"}{1}}{``z"}{2} = \upd{\emptyset}{``z"}{2}$.

Finally, the basic unit of reduction is a \textit{configuration}, defined $\set{Conf}_l = \set{term} \times \set{State}_l$.  We say that configuration $c=(t,\mu)$ is a value if $t$ is a value.

\subsection{Straightened Operational Semantics}

\label{sec:sos}

%

Rules in structural operational semantics are ordinarily written like logic programs, allowing them to be used to run programs both forward and backwards, and allowing premises to be proven in any order. However, in most rules, there are dependences between the variables that effectively permit only one ordering. In this section, we give an \textbf{alternate syntax} for small-step operational semantics rules, which makes this order explicit. This is essentially the conversion of the usual notation into A-normal form \cite{flanagan1993essence}. 

\begin{wrapfigure}{l}{0.4\textwidth}
\vspace{-0.7em}
\begin{mdframed}
\begin{small}
\begin{center}
\vspace{-0.2em}
\hspace{-0.4cm}\begin{tabular}{rcl}
\text{rule} & \gramdef & $c \sosto \text{rhs}$ \\
\text{rhs}  & \gramdef & $\build{c}$ \\
            & \gramalt & $\letstepto{c}{c}{\text{rhs}} $ \\
            & \gramalt & $\letcomp{c}{\set{semfun}(\overline{c})}{\text{rhs}}$ \\
\end{tabular}
\end{center}
\end{small}
\end{mdframed}
\vspace{-0.7em}
\caption{Notation for SOS}
\label{fig:sos-gram}
\vspace{-0.7em}
\end{wrapfigure}

The most immediate benefit of the Straightened Operational Semantics notation is that it orders the premises of a rule. One can also gain the ordering property by imposing a restriction on rules without a change in notation: Ibraheem and Schimdt's ``L-attributed semantics'' is exactly this, but for big-step semantics \cite{ibraheem1997adapting}. But there is an additional advantage of this new notation: it has an inductive structure, which allows defining recursive algorithms over rules.

\figref{fig:sos-gram} defines a grammar for SOS rules. These rules collectively define the step-to relation for a language $l$, $\sosto_l$. These rules are relations rather than functions, as they may fail and may have multiple matches. $R(A)$ denotes the image of a relation, i.e.: $R(x)$ refers to any $y$ such that $x \mathrel{R} y$.


A rule matches on a configuration, potentially binding several pattern variables. It then executes a right-hand side. Rule right-hand sides come in three alternatives. The two primary ones are that a rule's right-hand side may construct a new configuration from the bound variables, or may recursively invoke the step-to relation, matching the result to a new pattern. For example, the \textsc{AssnCong} rule from \secref{sec:contains-assn-rules} would be rendered as:

\vspace{-1em}
\begin{small}
\begin{mathpar}
(\assign{x}{e},\mu) \sosto \letstepto{(e,\mu)}{(e\p,\mu\p)}{(\assign{x}{e\p}, \mu\p)}
\end{mathpar}
\end{small}
\vspace{-1em}

\noindent As a third alternative, it may invoke an external semantic function. Semantic functions are meant to cover everything in an operational semantics that is not pure term rewriting, e.g.: arithmetic operations. Each language has its own set of allowed semantic functions.

\begin{definition}
Associated with each language $l$, there is a set of \textbf{semantic functions} $\set{semfun}_l$. Each is a relation\footnote{We define semantic ``functions" to actually be relations (i.e.: partial nondetermistic functions) so that this definition can be reused for abstract interpretation in \secref{sec:abstract-rewriting}.}
$R$ of type $\overline{\set{Conf}}_l \times \set{Conf}_l$, subject to the restriction that (1) if $\overline{c} \mathrel{R} d$, then $\overline{c}$ and $d$ are values, and (2) for each $\overline{c} \in \set{Conf}_l$, there are only finitely many $d$ such that $\overline{c} \mathrel{R} d$.
\end{definition}

For instance, there are two semantic functions for the \IMP language: $\set{semfun}_\IMP=\{+,<\}$. Both are partial functions, only defined on arguments of the form $((n_1,\mu_1),(n_2, \mu_2))$. Since these functions only act on their term arguments and ignore their $\mu$ arguments, we invoke them using the abbreviated notation

\begin{small}
\begin{mathpar}
\letcomp{n_3}{+(n_1,n_2)}{\text{rhs}}
\end{mathpar}
\end{small}

\noindent which is short for  $\letcomp{(n_3, \mu_3)}{+((n_1,\mu_1), (n_2, \mu_2))}{\text{rhs}}$, where $\mu_1,\mu_2$ are dummy values, and $\mu_3$ is an otherwise unused variable. 

Semantic functions give straightened-operational-semantics notation a lot of flexibility. They can be used to encode side-conditions, e.g.: ``$\letcomp{\true}{\text{isvalid}(x)}{c}$'' would fail to match if $x$ is not valid. They may include external effects such as I/O. As an example using semantic functions, the rule

\begin{small}
\begin{mathpar}
\infer[AddEval]{(v_1+v_2,\mu) \sosto (n,\mu)}{n = v_1+v_2}
\end{mathpar}
\end{small}

\noindent would be rendered as

\begin{small}
\begin{mathpar}
(v_1+v_2,\mu) \sosto \letcomp{n}{+(v_1,v_2)}{(n, \mu)}
\end{mathpar}
\end{small}

\noindent where the first occurrence of ``$+$'' refers to a nonval node of the object language, and the second occurrence refers to a semantic function of the meta-language. 

Overall, this notation gives a simple inductive structure to SOS rules. We will see in \secref{sec:sos-to-pam} how it makes the SOS-to-PAM conversion straightforward. And it loses very little generality from the conventional SOS notation: it essentially only assumes that the premises can be ordered. (And converting this notation back into the conventional form is easy: turn all RHS fragments into premises.)

\subsection{The Phased Abstract Machine}
\label{sec:pam}

In an SOS, a single step may involve many rules, and
each rule may perform multiple computations. The phased abstract
machine (PAM) breaks each of these into distinct steps. In doing so,
it simulates how a recursive functional program would interpret the
operational semantics.

\begin{figure}
\begin{mdframed}
\small{
\begin{center}
\begin{tabular}{rcl}
$\text{frame}$ & \gramdef & $\cxtfr{c}{\text{rhs}}$ \\
$K$ & \gramdef & $\halt$ \hfill{(Contexts)} \\
& $\sor$ & $K \frcomp \text{frame}$ \\
& $\sor$ & $k$ \hfill{(Context Vars)} \\
$\updown$ & \gramdef & $\up \;\sor\; \down$ \hfill{(Phase)} \\
$\text{pamState}$ & \gramdef & $\pam{c}{K}{\updown}$ \\
$\text{pamRhs}$ & \gramdef & $\text{pamState}$ \\
& $\sor$ & $\letcomp{c}{\set{semfun}(\overline{c})}{\text{pamRhs}}$ \\
$\text{pamRule}$ & \gramdef & $\text{pamState} \pamto \text{pamRhs}$
\end{tabular}
\end{center}
}
\end{mdframed}
\vspace{-0.1cm}
\caption{The Phased Abstract Machine}
\label{fig:pam-def}
\end{figure}

\begin{figure}
\begin{mdframed}
\small{
\begin{center}
\begin{tabular}{rclcl}
C & \gramdef & $\hole$ & $\sor$ & $\letcomp{c}{\set{semfun}(\overline{c})}{C}$
\end{tabular}
\end{center}
}
\end{mdframed}
\vspace{-0.15cm}
\caption{Abstract Machine: RHS contexts}
\label{fig:am-ctx-def}
\end{figure}

A PAM state takes the form $\pam{c}{K}{\updown}$. The configuration
$c$ is the same as for operational semantics. $K$ is a
\textit{context} or \textit{continuation}, which represents the remainder of an SOS right-hand
side. The phase is the
novel part. Each PAM state is either in the evaluating (``down'')
phase $\down$, or the returning (``up'') phase $\up$; an arbitrary phase is given the variable name ``$\updown$''.  A down state
$\pam{c}{K}{\down}$ can be interpreted as an intention for the PAM to
evaluate $c$ in a manner corresponding to one step of the operational
semantics, yielding $\pam{c\p}{K}{\up}$. In fact, in \appref{app:correctness}, we prove that a single step $c_1\sosto
c_2$ in the operational semantics perfectly corresponds to a sequence
of steps $\pam{c_1}{K}{\down} \pamto^* \pam{c_2}{K}{\up}$ in the PAM.

The full syntax of PAM rules is in \figref{fig:pam-def}. A PAM rule steps a left-hand state into a right-hand state,
potentially after invoking a sequence of semantic functions. A PAM state contains a configuration,
context, and phase. A context is a sequence of frames, terminating in $\halt$. 

Note that a \set{pamRhs} consists of a sequence of RHS fragments which terminate in a \set{pamState} $\pam{c}{K}{\updown}$. \figref{fig:am-ctx-def} captures this into a notation for RHS contexts, so that an arbitrary PAM rule may be written {\small $\pam{c_1}{K_1}{\updown_1} \pamto \plug{C}{\pam{c_2}{K_2}{\updown_2}}$}.
 
Finally, as alluded to in \secref{sec:overview-sos-to-am}, a frame like $\cxtfr{(t\p,\mu\p)}{(\assign{x}{t\p},\mu\p)}$ can be abbreviated to $\shortcxtfr{(\assign{x}{\hole_t},\hole_\mu)}$, since $t\p$ and $\mu\p$ are variables (i.e.: no destructuring).

A PAM rule $\pam{c_p}{K_p}{\updown_p}\pamto_l \text{rhs}_p$ for language $l$ is executed on a PAM state $\pam{c_1}{K_1}{\updown_1}$ as follows:

\begin{enumerate}
    \item Find a substitution $\sigma$ such that $\sigma(c_p)=c_1$, $\sigma(K_p)=K_1$.  Fail if no such $\sigma$ exists, or if $\updown_p \neq \updown_1$.
    \item Recursively evaluate $\text{rhs}_p$ as follows:
    \begin{itemize}
      \item If $\text{rhs}_p=\letcomp{c_\text{ret}}{\text{func}(\overline{c_\text{args}})}{\text{rhs}_p\p}$, with $\text{func} \in \set{semfun}_l$, pick $r\in \text{func}(\sigma(\overline{c_\text{args}}))$, and extend $\sigma$ to $\sigma\p$ s.t. $\sigma\p(c_\text{ret})=r$ and $\sigma\p(x)=\sigma(x)$ for all $x\in \text{dom}(\sigma)$. Fail if no such $\sigma\p$ exists. Then recursively evaluate $\text{rhs}_p\p$ on $\sigma\p$.
      \item If {\small $\text{rhs}_p=\pam{c_p\p}{K_p\p}{\updown_p\p}$}, return the new PAM state {\small $\pam{\sigma(c_p\p)}{\sigma(K_p\p)}{\updown_p\p}$}.
    \end{itemize}
\end{enumerate}

Let us give some example PAM rules. (An example execution will be in \secref{sec:has-example-pam-derivation}.) The \textsc{AssnCong} and \textsc{AssnEval} rules from \secref{sec:contains-assn-rules} get
transformed into the following three rules:

\vspace{-0.5em}
\begin{small}
\begin{mathpar}
  \begin{array}{r@{\ \pamto\ }l}
     \pam{(\assign{x}{e},\mu)}{k}{\down} &
                                           \pam{(e,\mu)}{k\frcomp \shortcxtfr{(\assign{x}{\hole_t},\hole_\mu)}}{\down} \\
     \pam{(t,\mu)}{k\frcomp \shortcxtfr{(\assign{x}{\hole_t},\hole_\mu)}}{\up} &
                                          \pam{(\assign{x}{t},\mu)}{k}{\up} \\

    \pam{(\assign{x}{v},\mu)}{k}{\down} & \pam{(\impskip, \upd{\mu}{x}{v})}{k}{\up} \\
   \end{array}
 \end{mathpar}
\end{small}
\vspace{-0.5em}

\noindent The \textsc{AssnCong} rule becomes a pair of mutually-inverse
PAM rules. One is a \textbf{down rule} which steps a down state to a
down state, signaling a recursive call. The other is an \textbf{up rule}, corresponding to a return. This is a distinguishing feature of congruence rules, and
is an important fact used when constructing the final abstract
machine.  Evaluation rules, on the other hand, typically become 
\textbf{down-up} rules. Note also that frames may be able to store
information: here, the frame $\shortcxtfr{(\assign{x}{\hole_t},\hole_\mu)}$ stores the variable to be assigned, $x$.

The PAM and operational semantics share a language's underlying semantic functions. The \textsc{AddEval} rule of, for instance,
\secref{sec:sos} becomes the following PAM rules:

\begin{small}
\begin{mathpar}
  \begin{array}{rcl}
  \pam{(v_1+v_2,\mu)}{k}{\down} &\pamto&\letcomp{n}{+(v_1,v_2)}
                                           {\pam{(n,\mu)}{k}{\down}}
    \\
    \pam{(v,\mu)}{k}{\down} &\pamto& \pam{(v,\mu)}{k}{\up} \\
    \end{array}
  \end{mathpar}
  \end{small}

\noindent The latter rule becomes redundant upon conversion to an abstract machine, which drops the phase.


\subsection{Abstract Machines}
\label{sec:am}

\begin{figure}
\begin{mdframed}
\small{
\begin{center}
\begin{tabular}{rcl}
$\text{amState}$ & \gramdef & $\am{c}{K}$ \\
$\text{amRhs}$ & \gramdef & $\text{amState}$  $\sor$  $\letcomp{c}{\set{semfun}(\overline{c})}{\text{amRhs}}$ \\
$\text{amRule}$ & \gramdef & $\text{amState} \amto \text{amRhs}$ 
\end{tabular}
\end{center}
}
\end{mdframed}
\vspace{-0.2cm}
\caption{Abstract Machines}
\label{fig:am-def}
\vspace{-0.25cm}
\end{figure}

Finally, the AM is similar to the PAM, except that an AM state does not contain a
phase. \figref{fig:am-def} gives a grammar for AM rules. We gave example rules for
the AM for \IMP in \secref{sec:contains-imp-machine}.

\paragraph{Summary of Notation} 

This section has introduced three versions of semantics, each defining their own transition relation.  Mnemonically, as the step-relation gets closer to the abstract machine, the arrow ``flattens out.'' They are:

\begin{enumerate}
\item $c_1 \sosto c_2$ (``squiggly arrow'') denotes one SOS step
   (\secref{sec:sos}).
\item $c_1 \pamto c_2$ (``hook arrow'') denotes one PAM step (\secref{sec:pam}).
\item $c_1 \amto c_2$ (``straight arrow'') denote one AM step (\secref{sec:am}).
\end{enumerate}


The conversion between PAM and AM introduces a
fourth system, the \textit{unfused abstract machine}, which is identical to the AM except that some
rules of the AM correspond to several rules of the
unfused AM. We thus do not distinguish it from the orthodox AM, except in one of the proofs, where its transition relation is given the symbol $\unfusedamto$ (``long arrow'').

\subsection{Splitting the SOS}
\label{sec:sos-to-pam}

\begin{figure*}
\footnotesize{
\begin{mdframed}\centering
\fbox{$\text{sosToPam}(\overline{\text{rules}})$}
\begin{align*}
\text{sosToPam}(\overline{\text{rules}}) & =\ \ \ 
                                                   \left(\bigcup_{r\in\overline{\text{rules}}} \text{sosRuleToPam}(r)\right) \\
& \ \ \ \ \ \ \ \  \cup \left\{ \pam{(t,s)}{\halt}{\up} \pamto \pam{(t,s)}{\halt}{\down}\right\} \\
& \ \ \ \where\ t,\, s\text{ are fresh variables}
\end{align*}
\\[0.8em]
\fbox{$\text{sosRuleToPam}(\text{rule})$}
\begin{align*}
\text{sosRuleToPam}(c \sosto \text{rhs}) & = & \text{sosRhsToPam}(\am{c}{k}{\down},
                                             k, \text{rhs})  & \hfill{(1)} \\
& & \ \ \ \where\ k\text{ is a fresh variable}  
\end{align*}
\\[0.8em]
\fbox{$\text{sosRhsToPam}(\text{pamState},K,\text{rhs})$}
\begin{align*}
\text{sosRhsToPam}(s, k, c) & = & \left\{ s\pamto \pam{c}{k}{\up} \right\} & \hfill{(2)} \\
\text{sosRhsToPam}(s, k, \letstepto{c_1}{c_2}{\text{rhs}}) & = &
                                                                 \left\{ s\pamto \pam{c_1}{k\p}{\down} \right\} & \hfill{(3)} \\
& & \ \ \ \ \ \ \ \ \mathrel{\cup} \text{sosRhsToPam}(\pam{c_2}{k\p}{\up},k,\text{rhs})   \\
& & \ \ \ \where\ k\p = k\frcomp \cxtfr{c_2}{\text{rhs}}  \\
\text{sosRhsToPam}(s, k, \letcomp{c_2}{\text{func}(\overline{c_1})}{\text{rhs}}) & = & \left\{ s \pamto \letcomp{c_2}{\text{func}(\overline{c_1})}{\pam{c_2}{k\p}{\down}} \right\} & \hfill{(4)} \\
& & \ \ \ \ \ \ \ \ \mathrel{\cup} \text{sosRhsToPam}(\pam{c_2}{k\p}{\down},k,\text{rhs})  \\
& & \ \ \ \where\ k\p = k\frcomp \cxtfr{c_2}{\text{rhs}} 
\end{align*}
\end{mdframed}
}
\vspace{-0.5em}
\caption{The SOS-to-PAM algorithm. Labels are used in the proofs of \appref{app:correctness}.}
\label{fig:sos-to-pam}
\vspace{-1em}
\end{figure*}

In this section, we present our algorithm for converting an
operational semantics to PAM. \figref{fig:sos-to-pam} defines the
$\textsc{sosToPam}$ function, which computes this transformation.

The algorithm generates PAM rules for each SOS rule.  For each SOS rule $c \sosto \text{rhs}$, it begins in
the state $\pam{c}{k}{\down}$, the start state for evaluation of
$c$. It then generates rules corresponding to each part of
$\text{rhs}$. For a semantic function, it transitions to a down state,
and begins the next rule in the same down state, so that they may
match in sequence. For a recursive invocation, it transitions to a
down state $\pam{c}{k\p}{\down}$, but begins the next rule in the
state $\pam{c}{k\p}{\up}$, so that other PAM rules must evaluate $c$
before proceeding with computations corresponding to this SOS rule. Finally, upon encountering
the end of the step $c$, it transitions to a state $\pam{c}{k}{\up}$,
returning $c$ up the stack.

For each step, it also pushes a
frame containing the remnant of the SOS $\text{rhs}$ onto the context,
both to help ensure rules may only match in the desired order, and
because the $\text{rhs}$ may contain variables bound in the left-hand
side, which must be preserved across rules.

After the algorithm finishes creating PAM rules for each of the SOS
rules, it adds one special rule, called the \textbf{reset rule}: \begin{small}
$\pam{(t,\mu)}{\halt}{\up} \pamto \pam{(t,\mu)}{\halt}{\down}$
\end{small}

The reset rule takes a state which corresponds to completing
one step of SOS evaluation, and changes the phase to $\down$ so that
evaluation may continue for another step. Note that it matches using a
nonvalue-variable $t$ so that it does not attempt to evaluate a term
after termination. Note also that the LHS and RHS of the reset rule
differ only in the phase. The translation from PAM to abstract machine hence removes this rule, as, upon dropping the phases, this rule would become a self-loop. It is also removed in the proof of Theorem \ref{thm:sos-pam}.

\subsection{Cutting PAM}

\label{sec:pam-to-am}

The PAM evaluates a term in lockstep with the original SOS rules. Yet, in both, each step of computation always begins at the root of the term, rather than jumping from one subterm to the next. By optimizing these extra steps away, our algorithm will create the abstract machine from the PAM.

Consider how the PAM evaluates the term $(1+(1+1))+1$, shown in
\figref{fig:ex-pam-reduction}. Notice how lines \ref{eqn:mirror-lines-first-start}--\ref{eqn:mirror-lines-first-end} mirror lines \ref{eqn:mirror-lines-second-start}--\ref{eqn:mirror-lines-second-end}. After evaluating $1+1$ deep within the term, the PAM walks up to the root, and then back down the same path. Knowing this, an optimized machine could jump directly from line \ref{eqn:mirror-lines-first-start} to \ref{eqn:after-mirror-lines}.

This insight is similar to the one behind Danvy's \textit{refocusing}, which converts a reduction semantics to an abstract machine \cite{danvy2004refocusing}. But in the setting of PAM, the necessary property becomes particularly simple and mechanical:

\begin{definition}
\label{def:up-rule-invertible}
An up-rule {\small $\pam{c_1}{K_1}{\up}\pamto \plug{C}{\pam{c_2}{K_2}{\up}}$} is \keyterm{invertible} if, for any $c_1$ nonvalue, {\small $\pam{c_2}{K_2}{\down} \pamto^* \pam{c_1}{K_1}{\down}$}.
\end{definition}

\label{sec:has-example-pam-derivation}
\begin{wrapfigure}{r}{0.48\textwidth}
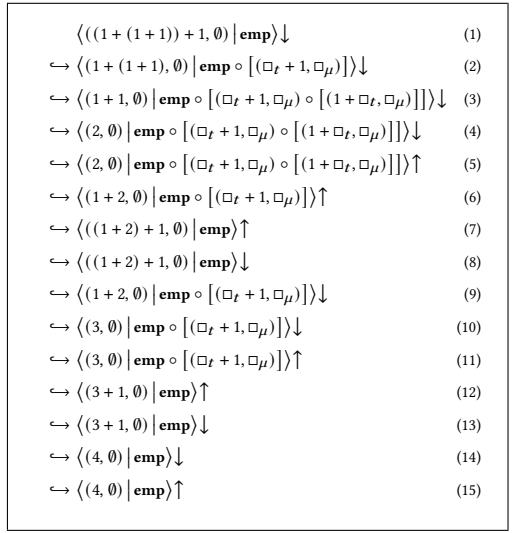

\begin{mdframed}\centering
\vspace{-0.3em}
{\tiny
\begin{align}
 & \pam{((1+(1+1))+1,\emptyenv)}{\halt}{\down} \\
\pamto & \pam{(1+(1+1), \emptyenv)}{\halt \frcomp
\shortcxtfr{(\hole_t+1,\hole_\mu)}}{\down} \\
\pamto & \pam{(1+1, \emptyenv)}{\halt \frcomp
                       \shortcxtfr{(\hole_t+1,\hole_\mu) \frcomp \shortcxtfr{(1+\hole_t,\hole_\mu)}}}{\down}  \\
\pamto & \pam{(2, \emptyenv)}{\halt \frcomp 
                \shortcxtfr{(\hole_t+1,\hole_\mu) \frcomp \shortcxtfr{(1+\hole_t,\hole_\mu)}}}{\down} \\
\pamto & \pam{(2, \emptyenv)}{\halt \frcomp 
           \shortcxtfr{(\hole_t+1,\hole_\mu) \frcomp \shortcxtfr{(1+\hole_t,\hole_\mu)}}}{\up} \\
\label{eqn:mirror-lines-first-start}
\pamto & \pam{(1+2, \emptyenv)}{\halt \frcomp 
             \shortcxtfr{(\hole_t+1,\hole_\mu)}}{\up} \\
\label{eqn:mirror-lines-first-end}
\pamto & \pam{((1+2)+1, \emptyenv)}{\halt}{\up} \\
\label{eqn:mirror-lines-second-start}
\pamto & \pam{((1+2)+1, \emptyenv)}{\halt}{\down} \\
\label{eqn:mirror-lines-second-end}
\pamto & \pam{(1+2, \emptyenv)}{\halt \frcomp 
                \shortcxtfr{(\hole_t+1,\hole_\mu)}}{\down} \\
\label{eqn:after-mirror-lines}
\pamto & \pam{(3, \emptyenv)}{\halt \frcomp 
             \shortcxtfr{(\hole_t+1,\hole_\mu)}}{\down} \\
\pamto & \pam{(3, \emptyenv)}{\halt \frcomp 
                \shortcxtfr{(\hole_t+1,\hole_\mu)}}{\up} \\
\pamto & \pam{(3+1, \emptyenv)}{\halt}{\up} \\
\pamto & \pam{(3+1, \emptyenv)}{\halt}{\down} \\
\pamto & \pam{(4, \emptyenv)}{\halt}{\down} \\
\pamto & \pam{(4, \emptyenv)}{\halt}{\up}
\end{align}
}
\vspace{-0.8em}
\end{mdframed}
\vspace{-0.5em}
\caption{Example PAM derivation.}
\label{fig:ex-pam-reduction}
\vspace{-1em}
\end{wrapfigure}

If an up-rule and its corresponding down-rules do not invoke any semantic functions, invertibility can be checked automatically via a reachability search. When all up-rules are invertible, we can show that $\pam{c}{K}{\up} \pamto^* \pam{c}{K}{\down}$ whenever $c$ is a non-value, meaning that these transitions may be skipped (Lemma \ref{lem:invertibility}). This means that all up-rules are redundant unless the LHS is a value, and hence they can be specialized to values. The phases now become redundant and can be removed, yielding the first abstract machine.

The last requirement for an abstract machine to be valid is not having any up-down rules, rules of the form $\pam{c}{k}{\up} \pamto C[\pam{c\p}{k\p}{\down}]$. An up-down rule follows from any SOS rule which simultaneously steps multiple subterms, and are not found in typical semantics. An example SOS rule which does result in an up-down rule is this lockstep-composition rule:

\begin{mathpar}
\infer[LockstepComp]
  {\lockstep{e_1}{e_2} \sosto \lockstep{e_1\p}{e_2\p}}
  {e_1 \sosto e_1\p & e_2 \sosto e_2\p}
\end{mathpar}

\noindent The \textsc{LockstepComp} rule differs from normal parallel composition, in that both components must step simultaneously. Up-down rules like this break the locality of the transition system, meaning that whether one subterm can make consecutive steps may depend on different parts of the tree. Correspondingly, it also means that a single step of the program may step multiple parts of the tree, making it difficult to have a meaningful notion of program counter.

Most of the time, the presence of an up-down rule will also cause some up-rules to not be invertible, making a prohibition on up-down rules redundant. However, there are pathological cases where this is not so. For example, consider the expression $\lockstep{e_1}{e_2}$ with the \textsc{LockstepComp} rule. The \textsc{LockstepComp} rule splits into 3 PAM rules, of which the third is an up-rule, $\pam{e_2\p}{k \frcomp \shortcxtfr{\lockstep{e_1\p}{\hole}}}{\up} \pamto \pam{\lockstep{e_1\p}{e_2\p}}{k}{\up}$. If it is possible for $e_1\p$ to be a value but not $e_2\p$, then this rule is not invertible.  However, if $e_1 \sosto e_1$ and $e_2 \sosto e_2$ for all $e_1, e_2$, then this rule is invertible. Hence, the PAM-to-AM algorithm includes an additional check that there are no up-down rules save the reset rule.

\paragraph{Algorithm: PAM to Unfused Abstract Machine}

\begin{enumerate}
\item Check that all up-rules for $l$ are invertible. Fail if not.
\item Check that there are no up-down rules other than the reset rule. Fail if not.
\item Remove the reset rule.
\item For each up-rule with LHS $\pam{(t,\mu)}{K}{\up}$, unify $t$ with a fresh value variable. The resulting $t\p$ will either have a value node at the root, or will consist of a single value variable. If $t$ fails to unify, remove this rule.
\item Remove all rules of the form $\pam{c}{K}{\down}\pamto \pam{c}{K}{\up}$, which would become self-loops.
\item Drop the phase $\updown$ from the $\text{pamState}$'s in all rules
\end{enumerate}

\paragraph{Fusing the Abstract Machine}

This unfused abstract machine still takes more intermediate steps than a normal abstract machine. The final abstract machine is created by \textit{fusing} successive rules together. A rule $\am{c_1}{K_1}\amto C_1[\am{c_1\p}{K_1\p}]$ is fused with a rule $\am{c_2}{K_2}\amto C_2[\am{c_2\p}{K_2\p}]$ by unifying $(c_1\p,K_1\p)$ with $(c_2,K_2)$, and replacing them with the new rule $\am{c_1}{K_1}\amto  C_1[C_2[\am{c_2\p}{K_2\p}]]$.

\begin{property}[Fusion]
\label{prop:fusion}
Consider two AM rules F and G, and let their fusion be FG. Then $\am{c}{K} \overset{F}{\amto} \am{c\p}{K\p} \overset{G}{\amto} \am{c\pp}{K\pp}$ if and only if $\am{c}{K} \overset{FG}{\amto} \am{c\pp}{K\pp}$.
\end{property}

There are two cases where rules should be fused. First, considering \figref{fig:sos-to-pam}, rules which invoke a semantic function always have only one possible successor rule, and should be fused. Without this, the abstract machine for \textsc{AddEval} would have an extra state for after it invokes the semantic computation $+(n_1,n_2)$, but before it plugs the result into a term. Second, up-rules should be fused with all possible successors. Without this, computing $e_1+e_2$ would have an extra state where, after evaluating $e_1$, it revisits $e_1+e_2$, rather than jumping straight into evaluating $e_2$. Both steps are strictly optional. However, doing so generates abstract machine rules which match the standard versions (as in e.g.: \cite{felleisen2009semantics}), and also generate more intuitive control-flow graphs.

For example, here are the final rules for assignment:

\todo{Explain caveats in camera ready <-- Wait, what caveats?}
\vspace{-1em}
\begin{mathpar}
  \begin{array}{r@{\ \amto\ }l}
     \am{(\assign{x}{e},\mu)}{k} &
                                           \am{(e,\mu)}{k\frcomp \shortcxtfr{(\assign{x}{\hole_t},\hole_\mu)}} \\
     \am{(v,\mu)}{k\frcomp \shortcxtfr{(\assign{x}{\hole_t},\hole_\mu)}} & \am{(\impskip, \upd{\mu}{x}{v})}{k} \\
   \end{array}
\end{mathpar}



\section{Correctness}
\label{sec:correctness-body}

This section provides the correspondence between the operational semantics and abstract machine. We present only the high-level theorems here, with the proofs available in 
\textbf{\appref{app:correctness}}.

The core idea of the correspondence is simple: The PAM emulates the SOS because each PAM rule was explicitly constructed to correspond to an RHS fragment of the SOS:

\begin{restatable}{theorem}{thmsospam}
\label{thm:sos-pam}
$c_1 \sosto_l^* c_2$ if and only if, for all contexts $K$, $\pam{c_1}{K}{\down} \pamto_l^* \pam{c_2}{K}{\up}$
\end{restatable}

\noindent The PAM and AM are equivalent because the AM merely removes some redundant steps from the PAM, and because the fused rules in the AM each correspond to several rules in the PAM. However, a PAM derivation may have some ``false starts" corresponding to a partially-applied SOS rule, and so we first must explain some technical definitions that determine which states are included in the correspondences.

\paragraph{\textbf{Stuck States}} The first kind of ``false start" comes from steps that cannot be completed.

\begin{definition}
A configuration/context pair $(c,K)$ is \textbf{non-stuck} if $\pam{c}{K}{\up} \pamto^* \pam{c\p}{\halt}{\up}$ for some $c\p$.
\end{definition}

Because each PAM rule corresponds to part of an SOS rule, our definition of non-stuckness is different from the usual one: it is intended to exclude terms which correspond to a partial match on an SOS rule. A single step $c_1 \sosto c_2$ in the SOS corresponds to a sequence $\pam{c_1}{\halt}{\down} \pamto^* \pam{c_2}{\halt}{\up}$ in the PAM, so a state is non-stuck if it can complete the current step. Stuck states result from SOS rules which only partially match a term. For example, the SOS rule

$$(\assign{a.b}{v}, \mu) \sosto \letcomp{(r,\mu\p)}{\text{Lookup}((a,\mu))}{\letcomp{\false}{\text{ContainsField}(r, b)}{(\error, \mu)}}$$

\noindent decomposes into 3 PAM rules. If $\text{Lookup}$ succeeds, the first brings $\pam{(\assign{a.b}{v}, \mu)}{K}{\down}$ into the state 

$$\pam{(r,\mu\p)}{K\frcomp \shortcxtfr{\letcomp{\false}{\text{ContainsField}(\hole_t,b)}{(\error, \mu)}}}{\down}$$ 

\noindent If $\false \neq \text{ContainsField}(r,b)$, then this will be a stuck state.

\paragraph{\textbf{Working Steps}}

As seen in the example in \secref{sec:pam-to-am}, many steps get removed when converting from PAM to AM. This causes the second form of ``false start."

\begin{definition}
An \keyterm{inversion sequence} beginning at $\pam{c}{K}{\up}$ is a sequence of transitions $\pam{c}{K}{\up} \pamto^* \pam{c}{K}{\down}$ which contains at most one application of the reset rule.
\end{definition}

This idea of an inversion sequence partitions a derivation $\pam{c_1}{K_1}{\down}\pamto^*\pam{c_2}{K_2}{\down}$ into two parts: the inversion sequences, which do redundant work, and the remainder, which we call the \keyterm{working steps}. A PAM state inside an inversion sequence might not correspond to any AM state.

\begin{definition}
A reduction $\pam{c_1}{K_1}{\updown_1}\pamto \pam{c_2}{K_2}{\updown_2}$ within a derivation is a \textbf{working step} if the derivation cannot be extended so that $\pam{c_1}{K_1}{\updown_1}$ is part of an inversion sequence.
\end{definition}

\paragraph{\textbf{PAM-AM Correspondence}} The PAM and AM correspond as follows, via the Unfused AM.

\begin{restatable}[PAM-Unfused AM: Forward]{theorem}{thmpamunfusedamforward}
\label{thm:pam-am}
Suppose, for some $c,K$, $\pam{c}{K}{\down}\pamto_l^* \pam{c\p}{K\p}{\down}$, $\pam{c\p}{K\p}{\down}$ is non-stuck, and the derivation's last step is working. Then there is a derivation $\am{c}{K}\unfusedamto_l^* \am{c\p}{K\p}$.
\end{restatable}

\begin{restatable}[PAM-Unfused AM: Backward]{theorem}{thmpamunfusedambackward}
\label{thm:am-pam}
If $\am{c}{K} \unfusedamto_l^* \am{c\p}{K\p}$, then there are phases $\updown_c$ and $\updown_{c\p}$ such that $\pam{c}{K}{\updown_c} \pamto_l^* \pam{c\p}{K\p}{\updown_{c\p}}$.
\end{restatable}

\noindent Reductions in the Unfused AM correspond to reductions in the normal AM unless the last rules used in the Unfused AM have been fused away.

\begin{restatable}[Unfused AM-AM]{theorem}{thmunfusedamam}
$\am{c}{K} \amto_l^* \am{c\p}{K\p}$ if and only if $\am{c}{K} \unfusedamto_l^* \am{c\p}{K\p}$ by a sequence of rules whose last rule is not fused away.
\end{restatable}

\section{Control-flow Graphs as Abstractions}
\label{sec:cfg-abstraction}

The abstract machine is a transition system describing all possible executions of a program. Applying an \textit{abstract interpretation} shrinks this to a finite one. Some abstractions yield a graph resembling a traditional CFG (\secref{sec:abstractions}).

Yet to compute on abstract states, one must also abstract the transition rules, and this traditionally requires manual definitions. Fortunately, we will find that the desired families of CFG can be obtained by a specific class of \textit{syntactic abstraction} which allow the transition rules to be abstracted automatically, via \textit{abstract rewriting} (\secref{sec:abstract-matching}--\secref{sec:abstract-rewriting}).
After constructing the abstract transition graph, more control can be
obtained by further combining states using a \textit{projection
  function} (\secref{sec:projections}). A final choice of control-flow
graph is then obtained by an abstraction/projection pair $(\alpha,
\pi)$.

The development of abstract-rewriting in this section is broadly
similar to that of \citet{bert1993abstract}, but departs
greatly in details to fit our formalism of terms and machines. In
\secref{sec:syntax-directed}, we explore connections to an
older technique called \textit{narrowing}.
 
\subsection{Abstract Terms, Abstract Matching}
\label{sec:abstract-matching}
 
Our goal is to find a notion of abstract terms flexible enough to allow us to express desired abstraction functions, but restricted enough that we can find a way to automatically apply the existing abstract machine rules to them. We accomplish this by defining a set of generalized terms $\set{term}\star$, satisfying $\set{term} \subset \set{term}\star$, where some nodes have been replaced with $\star$ nodes which represent any term.

\vspace{-1em}
\begin{mathpar}
\begin{tabular}{lcl}

$\set{term}\star$ & \gramdef & $\text{nonval}(\text{sym}, \overline{\text{term}\star})
                \mathrel{\gramalt} \text{val}(\text{sym}, \overline{\text{term}\star})$  \\
& &             $\mathrel{\gramalt} \text{const}
                \mathrel{\gramalt} x \mathrel{\gramalt} \star_\text{mt} $ \\
\end{tabular}

\end{mathpar}

\noindent Here, \text{mt} is a match type (\figref{fig:term-gram}), so that the allowed $\star$ nodes are $\valstar$, $\nonvalstar$, and $\allstar$. Formally, we define an ordering $\prec$ on $\set{term}\star$ as the reflexive, transitive, congruent closure of an ordering $\prec\p$, defined by the following relations for all $t, s, \overline{t}$:

\vspace{-1em}
\begin{small}
\begin{eqnarray*}
  \text{nonval}(s, \overline{t}) \prec\p \nonvalstar\ \ \    & \ \ \   \text{val}(s, \overline{t}) \prec\p \valstar\ \ \  &\ \ \    t \prec\p \allstar 
\end{eqnarray*}
\end{small}
\vspace{-1em}

\noindent A join operator $t_1 \sqcup t_2$ then follows as the least upper bound of $t_1$ and $t_2$. For instance, {
\small $(\assign{x}{1+1}) \sqcup (\assign{y}{2+1}) = (\assign{\valstar}{\valstar+1})$}. We can then define the set of concrete terms represented by $\abstr{t} \in \set{term}\star$ as:
 
\vspace{-0.5em}
\begin{small}
$$\gamma\left(\abstr{t}\right)= \left\{t \in \set{term} | t \prec  \abstr{t} \right\}$$
\end{small}
\vspace{-0.5em}

\noindent The power of this definition of $\set{term}\star$ is that it allows \textit{abstract matching}, which allows the rewriting machinery behind abstract machines to be automatically lifted to abstract terms.

\begin{definition}[Abstract Matching]
A pattern $t_p \in \set{term}$ matches an abstract term $\abstr{t}\in \set{term}\star$ if there is at least one $t\in \gamma(\abstr{t})$ and substitution $\sigma_t$ such that $\sigma_t(t_p)=t$. The witness of the abstract match is a substitution $\abstr{\sigma}$ defined:
\begin{small}
$$\abstr{\sigma}(x)=\bigsqcup \left\{\sigma_t(x) | t\in \gamma\left(\abstr{t}\right) \wedge \sigma_t(t_p)=t \right\}$$
\end{small}
\end{definition}

 For example, the abstract term $\allstar$ matches the pattern $v_1+v_2$ with a witness $\abstr{\sigma}$ with $\abstr{\sigma}(v_1)=\abstr{\sigma}(v_2)=\valstar$. We are now ready to state the main property of abstract matching.

\begin{property}[Abstract Matching (for terms)]
Let $t_p$ be a pattern, and $t\in\set{term}$ be a matching term, so that there is a substitution $\sigma$ with $\sigma(t_p)=t$. Consider a $t\p \in \set{term}\star$ such that $t\p \succ t$. Then $t\p$ matches $t_p$ with witness $\sigma\p$, where $\sigma\p$ satisfies $\sigma(x) \prec \sigma\p(x)$ for all $x\in \text{dom}(\sigma)=\text{dom}(\sigma\p)$, and $t \prec \sigma\p(t_p)$.
\end{property}

We assume there is some external definition of abstract reduction states $\abstr{\set{State}_l}$ (discussed in \secref{sec:abstract-rewriting}). After doing so, the definitions of abstract terms and states can be lifted to abstract configurations $\set{Conf}\star_l$, lists of abstract configurations $\overline{\set{Conf}\star}$, contexts $\set{Context}\star$, abstract machine states $\set{amState}\star$, and abstract semantic functions $\set{semfun}\star_l$, etc by transitively replacing all instances of $\set{term}$ and $\set{State}_l$ in their  definitions with $\set{term}\star$ and $\abstr{\set{State}_l}$. Abstract matching and the Abstract Matching Property are lifted likewise. To define abstract rewriting, we need a few more preliminaries.

\subsection{Abstract Rewriting}
\label{sec:abstract-rewriting}

Abstract rewriting works by taking the (P)AM execution algorithm of \secref{sec:pam}, and using abstract matching in place of regular matching. Doing so effectively simulates the possible executions of an abstract machine on a large set of terms. To define it, we must extend $\prec$ to other components of an abstract machine state.

First, we assume there is some externally-defined notion of abstract reduction states $\abstr{\set{State}_l}\supseteq \set{State}_l$ with ordering $\prec$. There must also be a notion of substitution satisfying $\sigma_1(\mu) \prec \sigma_2(\mu)$ for $\mu\in \abstr{\set{State}_l}$ if $\sigma_1(x) \prec \sigma_2(x)$ for all $x\in \text{dom}(\sigma_1)$. Finally, there must be an abstract matching procedure for states satisfying the Abstract Matching Property. In the common case where $\set{State}_l$ is the set of environments mapping variables to terms, this all follows by extending the normal definitions above to associative-commutative-idempotent terms.

We can now abstract matching and the $\prec$ ordering over abstract contexts $\set{Context}\star$, abstract configurations $\set{Conf}\star_l$, and lists of abstract configurations $\overline{\set{Conf}\star}_l$ via congruence. We extend $\prec$ over sets of configurations $\mathbb{P}(\set{Conf}\star_l)$ via the multiset ordering \cite{dershowitz1987termination}
\footnote{The multiset ordering, also known as the Dershowitz-Manna ordering, extends a base ordering $\prec_D$ on a elements $d \in D$ to an ordering $\prec_{PD}$ on multisets $A,B \in \mathbb{N}^D$ as follows: $A \prec_{PD} B$ if $A$ can be obtained from $B$ by applying the following operations any number of times: (1) removing an element, and (2) replacing a single element $b\in B$ by a finite set of elements $a_1, \dots, a_n$ such that each $a_i \prec_D b$.}
, and extend $\prec$ over semantic functions pointwise, i.e.: for $f_1, f_2 \in \set{semfun}_l$, $f_1 \prec f_2$ iff $f_1(x) \prec f_2(x)$ for all $x \in \overline{\set{Conf}\star}_l$.

Because normal AM execution may invoke an external semantic function, we need some way to abstract the result of semantic functions. We assume there is some externally-defined set $\abstr{\set{semfun}_l}$. Abstract rewriting will be hence parameterized over a ``base abstraction'' $\beta : \set{semfun}_l \rightarrow \abstr{\set{semfun}_l}$, satisfying the following property:

\vspace{-1em}
\begin{small}
$$\beta(f)(\abstr{\overline{c}}) \succ \bigcup \{ f(\overline{c}) | \overline{c} \in \gamma(\abstr{\overline{c}})  \} $$
\end{small}
\vspace{-1em}

\todo{There's a problem: Multiset ordering is sad if RHS is infinite. So, just instantiating this with abstract +, it does not hold....unless you define the base abstraction to pump out an infinite number of copies of star. Which, hey, why not?}

We are now ready to present abstract rewriting. An AM rule $\am{c_p}{K_p} \amto_l \text{rhs}_p$ for language $l$ is abstractly executed on an AM state $\am{c_1}{K_1}$ using base abstraction $\beta$ as follows:

\begin{enumerate}
    \item Compute the abstract match of $\am{c_p}{K_p}$ and $\am{c_1}{K_1}$, giving witness $\abstr{\sigma}$; fail if they do not abstractly match.
    \item Recursively evaluate $\text{rhs}_p$ as follows:
    \begin{itemize}
      \item If $\text{rhs}_p=\letcomp{c_\text{ret}}{\text{func}(\overline{c_\text{args}})}{\text{rhs}_p\p}$, with $\text{func} \in \set{semfun}_l$, then pick $r\in \beta(\text{func})(\abstr{\sigma}(\overline{c_\text{args}}))$ and compute $\abstr{\sigma_r}$ as the witness of abstractly matching $c_\text{ret}$ against $r$. Define $\abstr{\sigma\p}$ by $\abstr{\sigma\p}(x)=\abstr{\sigma}(x)$ for $x\in\text{dom}(\abstr{\sigma})$, and $\abstr{\sigma\p}(y)=\abstr{\sigma_\text{r}}(y)$ for $y\in \text{dom}(\abstr{\sigma_\text{r}})$. Fail if no such $\abstr{\sigma\p}$ exists. Then recursively evaluate $\text{rhs}_p\p$ using $\abstr{\sigma\p}$ as the new witness.
      \item If $\text{rhs}_p=\am{c_p\p}{K_p\p}$, return the new abstract AM state $\am{\abstr{\sigma}(c_p\p)}{\abstr{\sigma}(K_p\p)}$.
    \end{itemize}
\end{enumerate}

 \noindent If $\am{\abstr{c_1}}{\abstr{K_1}}$ steps to {\small $\am{\abstr{c_2}}{\abstr{K_2}}$} by abstractly executing an AM rule with base abstraction $\beta$, we say that {\small $\am{\abstr{c_1}}{\abstr{K_1}} \absred{\beta}  \am{\abstr{c_2}}{\abstr{K_2}}$}. Here is the fundamental property relating abstract and concrete rewriting:

\begin{lemma}[Lifting Lemma]
If {\small $\am{c_1}{K_1} \prec \am{\abstr{c_1}}{\abstr{K_1}}$}, and {\small $\am{c_1}{K_1} \amto \am{c_2}{K_2}$} by rule F, then, for any $\beta$, there is a {\small $\am{\abstr{c_2}}{\abstr{K_2}}$} such that {\small $\am{c_2}{K_2} \prec \am{\abstr{c_2}}{\abstr{K_2}}$} and {\small $\am{\abstr{c_1}}{\abstr{K_1}} \absred{\beta}  \am{\abstr{c_2}}{\abstr{K_2}}$} by  $F$.
\end{lemma}

Note that, if $\beta$ is the identity function, then $\absred{\beta}$ is the same as $\amto$. Hence, the Lifting Lemma theorem follows from a more general statement.

\begin{restatable}[Generalized Lifting Lemma]{lemma}{thmgenabsrewriting}
\label{thm:gen-abs-rewriting}
Let $\beta_1, \beta_2$ be base abstractions where $\beta_1$ is pointwise less than $\beta_2$, i.e.: $\beta_1(f)(\overline{c}) \prec \beta_2(f)(\overline{c})$ for all $f,c$. Suppose $\am{c_1}{K_1} \prec \am{\abstr{c_1}}{\abstr{K_1}}$, and also $\am{c_1}{K_1} \absred{\beta_1} \am{c_2}{K_2}$ by rule F. Then there is a $\am{\abstr{c_2}}{\abstr{K_2}}$ such that $\am{c_2}{K_2} \prec \am{\abstr{c_2}}{\abstr{K_2}}$ and $\am{\abstr{c_1}}{\abstr{K_1}} \absred{\beta_2}  \am{\abstr{c_2}}{\abstr{K_2}}$ by rule $F$.
\end{restatable}
\begin{proof}
See \appref{app:abs-rewriting-proofs}.
\end{proof}

\subsection{Machine Abstractions}
\label{sec:abstractions}

This section finally ties the knot on the adage ``CFGs are an abstraction of control-flow'' by defining the relation between a program's actual control flow (the concrete state transition graph) and its finite CFG. Intuitively, this relation involves skipping intermediate steps and combining states that correspond to the same program point. Yet getting the details right is tricky.

The simple function call \code{f()} exemplifies the difficulties of defining this relationship. Running \code{f()} may evaluate arbitrary code. (In the language of abstract rewriting, without context on \code{f}, it evaluates to an arbitrary program $\nonvalstar$.) An intraprocedural CFG generator must generate a small fixed number of nodes (typically, $1$ evaluation node or part of a basic-block node) to represent the call to \code{f()}; any additional nodes depending on the definition of \code{f} makes it no longer intraprocedural. It must then choose to either (1) draw edges through these nodes continuing onwards, or (2) when \code{f()} provably does not terminate, it may optionally draw edges into but not out of these nodes. Under a naive ``combine states and skip steps'' definition of valid abstraction, an intraprocedural CFG-generator must first perform an interprocedural termination analysis to be sound. And if \code{f()} only conditionally diverges, then both options are incorrect.

We shall present a relation which does indeed justify abstracting \code{f()} to $\valstar$, using a definition with subtle treatment of nontermination and branching. The upshot of this work is the ability to write control-flow abstractions with formal guarantees whose actual implementation is trivial; abstracting function calls as described here is but a small modification to the $5$ lines of \figref{fig:val-irrelevance-code}. And, in spite of the definition's complexity, actually showing a candidate abstraction meets the definition is usually trivial.

Over the next paragraphs, we develop the abstraction preorder $a \sqsubseteq b$ between abstract states. From this, we obtain our first formal definition of a CFG: when $G$ is the concrete transition graph for some program $P$, and $H$ is a finite abstract transition graph, every state in $G$ is $\sqsubseteq$ some state in $H$, and every state in $H$ is $\sqsupseteq$ some state in $G$, then $H$ is a valid CFG for $P$.

The $(\sqsubseteq)$ relation is built from three smaller relations. Clearly, $(\sqsubseteq)$ must contain the $(\prec)$ ordering, so that a single CFG node may describe concrete states that differ only in values. It should be able to ignore some steps of computation (e.g.: desugaring a while-loop), and so should involve {\small $(\absred{\beta})$}. It also must be able to skip over loops regardless of termination. We define the $(\triangleleft)$ relation to skip over infinite loops. Intuitively, $a=\am{(\abstr{t}, \abstr{\mu})}{\abstr{K}} \triangleleft \am{(\valstar, \top)}{\abstr{K}'}$ if $\top$ is the ``any'' state (i.e.: overapproximates all possible effects) and all executions of $a$ get ``trapped'' in some part of the program, with $\abstr{K}'$ never getting popped off the stack.

\todo{Wait why am I allowing breaking out here? Isn't that redundant with the other parts of the abstraction definition? Also, why separate K and K'; I don't think this works unless K=K'.  For camera ready, need to see if this can be simplified....as of 5/30/2022, I don't understand this comment.}

\begin{definition}[Nontermination-cutting ordering]
Let $\top_l$ be a maximal element of $\abstr{\set{State}_l}$.  Consider a state $a=\am{(\abstr{t}, \abstr{s})}{\abstr{K}}\in\set{amState}\star$, and let $\abstr{K}\p$ be a subcontext of $\abstr{K}$. Suppose that, for all $\am{(t,s)}{K}\prec a$, and for all derivations of the form $\am{(t,s)}{K}\amto^* \am{(t\p,s\p)}{K\p}$, either $K$ is a subterm of $K\p$, or there is a subderivation of the form $\am{(t,s)}{K}\amto^* \am{(t\pp, s\pp)}{K}$ such that $t\pp \prec \valstar$. Then $a \triangleleft \am{(\valstar,\top_l)}{\abstr{K}\p}$. If $a$ does not satisfy this condition, then $\nexists a\p. a \triangleleft a\p$.
\end{definition}

We now combine the $(\prec), (\triangleleft),$ and {\small $(\absred{\beta})$} orderings to fully describe what an abstraction may do. The naive approach would be to take the transitive closure of $(\prec) \cup (\triangleleft) \cup (\absred{\beta})$. But this would permit abstracting {\small $\am{\ite{\valstar}{A}{B}}{K}$} to {\small $\am{A}{K}$}! Instead:

\begin{definition}[Ordering $\sqsubseteq$ of $\set{amState}\star$]
\label{defn:abstraction-ordering}
The relation $\sqsubseteq$ is defined inductively as follows: $a\sqsubseteq b$ if any of the following hold:
\begin{enumerate}
    \item $a = b$
    \item For some $c$, $a \prec c$ and $c \sqsubseteq b$
    \item For some $c$,  $a \triangleleft c$ and $c \sqsubseteq b$
    \item \textit{For all} $c$ such that $a \absred{\beta} c$, $c \sqsubseteq b$
\end{enumerate}
\end{definition}

We can now define an abstract machine abstraction to be a pair $(\alpha, \beta)$, where $\beta$ is a base abstraction and $\alpha : \set{amState}\star \rightarrow \set{amState}\star$ is a function which is an upper closure operator under the $\sqsubseteq$ ordering, meaning it is monotone and satisfies $x \sqsubseteq \alpha(x)$ and $\alpha(\alpha(x)) = \alpha(x)$. It is well-known that such an upper closure operator establishes a Galois connection between $\set{amState}\star$ and the image of $\alpha$, $\alpha(\set{amState}\star)$ \cite{nielson2015principles}. We will assume that every $\alpha$ is associated with a unique $\beta$, and will abbreviate the machine abstraction $(\alpha, \beta)$ as just $\alpha$. \appref{app:graph-pattern-correctness} adds a few additional technical restrictions on $\alpha$. These restrictions have no bearing on interpreted-mode graph generation, but do rule out some pathological cases in the correctness proofs for compiled-mode graph generation.

We now define the abstract transition relation $\absred{\alpha}$ for $\alpha$ as: if {\small $a \absred{\beta}  b$}, then {\small $a \absred{\alpha} \alpha(b)$}. In other words, the abstract transition relation is the same as abstract rewriting, except that the RHS of a transition must always lie within the image of the abstraction $\alpha$.
We now state the fundamental theorem of abstract transitions, proved in \appref{app:abs-rewriting-proofs}.

\begin{restatable}[Abstract Transition]{theorem}{thmabstransition}
For $a, b \in \set{amState}$, if $\alpha$ is an abstraction with base abstraction $\beta$, and $a\amto b$, then either $b \sqsubseteq \alpha(a)$, or $\exists g \in \set{amState}\star.\; \alpha(a) \absred{\alpha} g \wedge b \sqsubseteq g$.
\end{restatable}

\noindent Following are some example abstractions.

\paragraph{Abstraction: Value-Irrelevance}

The value-irrelevance abstraction (code in \figref{fig:val-irrelevance-code}) maps each node
$\text{val}(\text{sym},t)$ and each constant to $\valstar$, and each
semantic function to the constant $x\mapsto \valstar$. Combining this
with the abstract machine for \IMP yields an expression-level
control-flow graph, as in \figref{fig:cfg-variants-expr}. \textsc{Tiger} and \textsc{MITScript} use a modified version described at the beginning of this section, which also abstracts all function calls to $\valstar$.


\paragraph{Abstraction: Expression-Irrelevance}

This abstraction is like value-irrelevance, but
it also ``skips'' the evaluation of expressions by mapping any expression
under focus to $\valstar$. In doing so, it overapproximates modifications to the state from running $e$; the easiest implementation is to add the mapping $[\valstar\mapsto \valstar]$. Combining this with an abstract machine yields a statement-level control-flow graph.

\paragraph{The Boolean-Tracking Abstraction}

This abstraction is similar to  value-irrelevance, except that it preserves some $\true$ and $\false$ values. It differs in two ways: (1) boolean-valued semantic functions such as $<$ nondeterministically return $\{\true,\false\}$, (2) for a configuration $(t,\mu)$, it preserves all $\true$ and $\false$ values in $t$, as well as the value of $\mu(v)$ for each $v$ in a provided set of \textit{tracking variables} $V$. Including the variable ``b'' in the tracked set, and combining this with the basic-block projection (\secref{sec:projections}) yields the path-sensitive control-flow graph seen in \figref{fig:cfg-variants-path-sensitive}.


\subsection{Projections}
\label{sec:projections}

A \keyterm{projection}, also called a \keyterm{quotient map}, is a function $\pi : \set{amState}\star \rightarrow \set{amState}\star$. They are used after constructing the initial CFG to merge together extra nodes.

The definition below comes from \S 5.4 of \citet{manolios2001mechanical}, which presented it in the context of bisimulations. It differs slightly from the standard graph-theoretic definitions, in that it has an extra condition to prevent spurious self-loops.

\begin{definition}
\label{defn:quotient-graph}
Let $(V,E)$ be a graph with $V \subseteq \set{amState}\star$, $E \subseteq \left(\set{amState}\star^2\right)$. Let $\pi$ be a projection. Then the \textbf{projected graph} (or \textbf{quotient graph}) is the graph $(V\p,E\p)$ satisfying:

\begin{enumerate} 
\item $V\p=\pi(V)$
\item For $a \neq b$, $(a,b)\in E\p$ iff there is $(c,d)\in E$ such that $(a,b) = (\pi(c), \pi(d))$.
\item $(a, a) \in E\p$ iff, for all $b\in V$ such that $\pi(b)=a$, there is $c\in V$ such that $(b, c) \in E$.
\end{enumerate}
\end{definition}

\noindent Many of the uses of projections could be accomplished by instead using a coarser abstraction. However, projections have the advantage that they have no additional requirements to prove: they can be any arbitrary function. The addition of projections gives our final definition of a CFG: when $G$ is the concrete transition graph for some program $P$, and $H$ is a finite abstract transition graph, and every state in $G$ is $\sqsubseteq$ some state in $H$ and vice versa, then for any projection $\pi$, the projected graph of $H$ under $\pi$ is a valid CFG for $P$. In short, a CFG is \textbf{a projection of the transition graph of abstracted abstract machine states}.


Projections can be defined either manually, or automatically by the
graph-pattern code generator (\secref{sec:syntax-directed}), and are most often used to hide internal details of a language's semantics, such as in the graph pattern of \figref{fig:front-page-result}, which merges away the internal steps of a while-loop which are mere artifacts of the SOS rules. But one important projection goes further:

\paragraph{Projection: Basic Block}

\todo{This is kinda wrong}
The basic-block projection inputs $\am{c}{K}$ and removes all but the last top-level sequence-nodes from $c$ and $K$, essentially identifying each statement of a basic-block with the last statement in the block. In combination with the expression-irrelevance abstraction, this yields the classic basic-block control-flow graph, as in \figref{fig:cfg-variants-bb}.

\todo{Put this somewhere:
The correspondence between operational semantics and control-flow graphs is then given by Definition \ref{defn:quotient-graph} and Theorems \ref{thm:sos-pam}, \ref{thm:pam-am},  \ref{thm:am-pam}, and \ref{thm:trans-abstraction}.
}

\subsection{Termination}
\label{sec:termination}

Will the interpreted-mode CFG-generation algorithm terminate in a finite control-flow graph? If the abstraction used is the identity, and the input program may have infinitely many concrete states, the answer is a clear no. If the abstraction reduces everything to $\valstar$, the answer is a clear yes. In general, it depends both on the abstraction as well as the rules of the language.

If CFG-generation terminates for a program P, that means there are
only finitely-many states reachable under the $\absred{\alpha}$
relation from P's start state. Term-rewriting researchers call this
property ``global finiteness,'' and have proven it is usually equivalent to another property,
``quasi-termination'' \cite{dershowitz1987termination}. While the
literature on these properties can help, there must still be a
separate proof for the termination of CFG -generation for every
language/abstraction pair. Such a proof may be a
tedious one of showing that e.g.: every while-loop steps to
$\ite{e}{\seq{s}{\while{e}{s}}}{\impskip}$, which eventually steps
back to $\while{e}{s}$, and never grows the stack.

Fortunately, for abstractions which discard context, the graph-pattern
generation algorithm of \secref{sec:syntax-directed} does this
analysis automatically. Because the transitions discovered by abstract
rewriting for a specific program are a subset of the union of graph
patterns for all AST nodes in that program, we discuss in \secref{sec:automated-termination} and prove in \appref{app:graph-pattern-correctness} that, if graph-pattern
generation for a given language and abstraction terminates, then so
does interpreted-mode CFG generation for all programs.

\section{Syntax-Directed CFG Generators}
\label{sec:syntax-directed}

Although it may sound like a large leap to go from generating CFGs for
specific programs to statically inferring the control-flow of every
possible node, it is actually only little more than running the
CFG-generation algorithm of \secref{sec:cfg-abstraction} on a term
with variables. Indeed, the core implementation in \mandate is only 22
lines of code. Hence, the description in
\secref{sec:overview-graph-pat} was already mostly complete, and we
have only a few details to add. Correctness results are given in \appref{app:graph-pattern-correctness}.

There are two primary points of simplification in
\secref{sec:overview-graph-pat}. First, it did not explain
how to run AM rules on terms with variables. Second, it ignored match types.

\paragraph{Executing Terms with Variables}

A $\star$ term does not have identity. In a state $\am{(\nonvalstar,\mu)}{K}$, some work is needed to determine which (if any) term in the original program that $\nonvalstar$ corresponds to. So abstract rewriting can discover that $\nonvalstar + \nonvalstar$ steps
to a state that executes a $\nonvalstar$, but it cannot tell you which
$\nonvalstar$ it executes first, which is needed to tie a control-flow graph to the original AST.

Variables, in contrast, do have identity. And variables can serve in the place of $\star$, for there is a simple way to determine if an AM state with variables could be instantiated to match
the LHS of an AM rule: if they unify
.
The result is then the RHS of
the rule with the substitution applied. This shows that, from the term
$a_\mtnonval+b_\mtnonval$, $a_\mtnonval$ is evaluated first.

This operation is called \textit{narrowing}, used since 1975 in decision procedures \cite{lankford1975canonical} and
functional-logic programming. We present a variant of narrowing which makes it suitable for abstract interpretation of semantics: we say that $f \narrowsto g$ if, for some rule $x \rightarrow y$, $f$ and $x$ unify by substitution $\sigma$, and $g=\sigma(y)$. The resulting relation $\absnarrow{\alpha}$ is defined
identically to the development of $\absred{\alpha}$ given in
\secref{sec:abstract-rewriting} and \secref{sec:abstractions}, except that the witness
$\abstr{\sigma}$ is computed by unification instead of by matching.

Note that the abstract-rewriting of
\secref{sec:cfg-abstraction} can be viewed as an overapproximation of our version of narrowing that follows each
narrowing step with an abstraction that replaces each occurrence of the same with distinct fresh variables (i.e.: $\star$ nodes), along
with extensions to handle match types and semantic functions.

Our abstract rewriting differs from conventional narrowing in
an important way. Conventional narrowing is actually a ternary relation. For instance,
the rule $f(f(x))\rightarrow x$ enables the derivation
$f(y)\narrowsto_{[y\mapsto f(y\p)]} y\p$, with the
unifying substitution as the third component of the relation. We turn
it into a binary relation by applying the substitution on the right,
and ignoring it on the left. 

This small tweak changes the interpretation of narrowing while
preserving its ability to overapproximate rewriting. In conventional narrowing, $[v_1 \mapsto v_2]$ represents
an environment with a single, to-be-determined key/value pair. $v_1$ may unify with both the names ``x'' and ``y'', but only in
parallel universes. In abstract rewriting, $[v_1\mapsto v_2]$
represents arbitrary environments, where any name maps to any value.

\todo{Narrowing becomes interesting when you chain together several transitions. The ``parallel'' universes refers to how, if A ~>sigma1 B, and B ~>sigma2 C, then  must have sigma1=sigma2 in order to compose them. Different paths yield different compositions. If A ~>sigma1 D and A~>sigma3 E, and sigma1 and sigma3 are incompatible, it matters not if D and E can both step to F; the paths will not rejoin. Perhaps I can elaborate this somewhere.}

\paragraph{Graph-Pattern Generation (Now with Match-Types)}

The final requirement to generate graph patterns for a language $l$ is
that the ordering for $\set{State}_l$ has a maximum value $\top_l$. Then, graph
patterns are generated as follows: For each
node type $N$, generate
the abstract transition graph by narrowing from the start state
{\small $S=\am{(N(\overline{x^i_\mtnonval}), \top_l)}{k}$}, where the $x^i$ are
arbitrary non-value variables, and $k$ is fresh. Any time a state of the form {\small $\am{(e_\mtnonval,\mu)}{K}$} is
encountered, instead of narrowing, add a transitive edge to
{\small $\am{(\valstar, \top_l)}{K}$}. Halt at any state $\am{(v, \mu)}{k}$,
where $k$ is the same context variable as $S$, and $v$, $\mu$ are
an arbitrary value and reduction state.

From this, we see that the graph pattern in \figref{fig:front-page-result} was
slightly simplified. The real one has each starting variable
annotated with $\mtnonval$ and replaces each $\star$ node with $\valstar$.

\paragraph{Code-Generation}

Given a graph pattern, the code generator will perform a greedy search to find a projection which groups the graph nodes into loop-free segments, each group associated with an AST node. From this graph, the final code generation step is trivial, with one \code{connect} statement per edge in the projected graph pattern (see \figref{fig:front-page-result}). The discovered projection is guaranteed to be valid, and hence the generated CFG-generator is guaranteed to be correct. While they do not appear in the languages under study, there are also pathological cases where a projection of the desired form does not exist; in this case, the search terminates with failure.

The code generator's high-level workings were in \secref{sec:overview-graph-pat}; some additional details are in
\appref{app:code-gen}; examples are available in \secref{sec:tiger-mitscript-cfgs} and the supplementary material.

\subsection{An Automated Termination-Prover}
\label{sec:automated-termination}

As the normal abstract state transition graph overapproximates all concrete executions of a program, a graph pattern for a language construct overapproximates the relevant fragment of all abstract state transition graphs.

Is it possible to have finite graph patterns but infinite abstract transition graphs, i.e.: for the compiled-mode CFG generator to terminate, but not the interpreted-mode one? With some light assumptions, we can show this is not the case. Hence, while the output of the compiled-mode CFG-generator will always be less precise than that of an interpreted-mode generator, it turns out running the compiled-mode generator once per language will prove that the interpreted-mode generator terminates on all programs in that language.

\begin{restatable}{theorem}{thmbothfinite}
\label{thm:both-finite}
Let $a\in\set{amState}_l$ and $\alpha$ be a machine abstraction.  If the graph patterns under abstraction $\alpha$ for all nodes in $a$ are finite, then only finitely many states are reachable from $a$ under the $\absred{\alpha}$ relation.
\end{restatable}

This proof requires a few more technical assumptions and a refinement to the definition of $\absnarrow{\alpha}$. The details are in \appref{app:graph-pattern-correctness}.

\section{Deriving Control from a \mandate}
\label{sec:mandate}

\todo{Give a review of how much we've done so far, and then say what's left}

\begin{wrapfigure}{r}{0.55\textwidth}
\begin{mdframed}
\vspace{-0.2em}
\begin{center}
\hspace{-05cm}\begin{lstlisting}[language=Haskell,basicstyle=\tiny\sffamily]
name "assn-cong" $
mkRule5 (\x e e' mu mu' ->
 let (gx, ne, ge') = (GVar x, NVar e, GVar e')
 in StepTo (conf (Assign gx ne) mu)
     (LetStepTo (conf ge' mu') (conf ne mu)
      (Build $ conf (Assign gx ge') mu')))
\end{lstlisting}
\end{center}
\vspace{-0.2em}
\end{mdframed}
\vspace{-0.8em}
\caption{Encoding of the \textsc{AssnCong} rule from \secref{sec:sos}}
\vspace{-0.5em}
\label{fig:sos-dsl}
\end{wrapfigure}

We have implemented our approach in a tool called \mandate. \mandate
takes as input an operational semantics for a language as an embedded
Haskell DSL, and generates a control-flow graph generator for that
language for every abstraction/projection supplied. It can then output
a generated CFG to the graph-visualization language DOT.  \mandate
totals approximately $9600$ lines of Haskell: $4100$ lines in the core engine, and $5500$ lines
for our language definitions and example analyzers. $1350$ of those lines define the $80$ SOS rules for \textsc{Tiger} and $60$ rules for \textsc{MITScript}, using the DSL depicted in \figref{fig:sos-dsl}. $550$ of those lines are automatically-generated CFG-generation code.


\begin{wraptable}[12]{l}{0.55\textwidth}
\vspace{-1.1em}
\caption{}
\vspace{-1em}
\label{table:cfg-impls}
\begin{tabular}{r|ccccc}
   \hline
     & \multicolumn{3}{c}{Interpreted} & \multicolumn{2}{c}{Compiled} \\
          & E & S & P & E & S \\
\hline
  \textbf{\IMP} & \checkmark & \checkmark & \checkmark & \checkmark & \checkmark \\
  \textbf{\textsc{Tiger}} & \checkmark & N/A & $\times$ & \checkmark & N/A \\
  \textbf{\textsc{MITScript}} & \checkmark & \checkmark & $\times$ & \checkmark & \checkmark \\
\hline
\end{tabular}
\vspace{0.5em}

\caption{Example analyzers}
\vspace{-0.5em}
\label{table:example-analyzers}
\begin{tabular}{r|r|ccc}
   \hline
          & LOC & \IMP & \textsc{Tiger} & \textsc{MIT} \\
\hline
  \textbf{Constant-prop} & $115$ & \checkmark & \checkmark & \checkmark \\
  \textbf{Paren-balancing} & $49$ & \checkmark & $\times$ & $\times$ \\
\hline
\end{tabular}
\end{wraptable}

Table \ref{table:cfg-impls} lists the CFG-generators we have generated
using \mandate. The columns E, S, and P correspond to the
expression-level, statement-level, and path-sensitive CFGs from
\secref{sec:introduction}, and which are generated by the
three abstractions from \secref{sec:abstractions}. The expression- and
statement-level CFG-generators come in interpreted-mode and
compiled-mode flavors. \textsc{Tiger} lacks a statement-level CFG-generator, because everything in \textsc{Tiger} is an expression. \secref{sec:tiger-mitscript-cfgs} explains the structure of heaps in \textsc{Tiger} and \textsc{MITScript}, and the implementation that would be required for \mandate to support the boolean-tracking abstraction for them.

The high readability of the generated CFG-generators (included in the supplementary material) allows for easy inspection and comparison to intuitive control-flow. But, to further test the usefulness of the generated CFGs, we built two example analyzers, summarized in Table \ref{table:example-analyzers}. The first is a simple constant-propagation analysis on expression-level CFGs, supporting assignments and integer arithmetic. The second is the parenthesis-balancing analyzer described in \secref{sec:introduction}, built atop the path-sensitive CFG. Even though \mandate was built as a demonstration of theory rather than as a practical tool, the simplicity of this exercise is further evidence that \mandate's output does indeed correspond to conventional hand-written CFG-generators, while their brevity reinforces our thesis that \textbf{having the appropriate kind of CFG-generator greatly simplifies tool construction}.

In the remainder of this section, we demonstrate the power of Mandate by showing how it
generates concise, readable code even in the face of complicated language constructs.







\subsection{Control-Flow Graphs for \textsc{Tiger} and \textsc{MITScript}}
\label{sec:tiger-mitscript-cfgs}

\begin{wrapfigure}[15]{r}{0.5\textwidth}
\begin{center}
\vspace{-0.5cm}
\begin{lstlisting}[language=Haskell,basicstyle=\tiny\sffamily]
(ConsFrame (HeapAddr 0) NilFrame,
 JustSimpMap $ SimpEnvMap $ Map.fromList
  [ (HeapAddr 0,
     ReducedRec
     $ RedRecCons (RedRecPair (Name "print")
                   (RefVal $ HeapAddr 1))
     $ RedRecCons (RedRecPair (Name "read")
                   (RefVal $ HeapAddr 2))
     $ RedRecCons (RedRecPair (Name "intcast")
                   (RefVal $ HeapAddr 3))
     $ Parent $ HeapAddr $ -1)
  , (HeapAddr 1, builtinPrint)
  , (HeapAddr 2, builtinRead)
  , (HeapAddr 3, builtinIntCast)
  ])                                  
\end{lstlisting}
\end{center}
\vspace{-1em}
\caption{Starting state of \textsc{MITScript} programs}
\label{fig:mitscript-heap}
\end{wrapfigure}

Previous sections used \IMP as the running language, which, in our implementation, has only $20$ SOS rules, with low complexity. In this section, we explain how our techniques work when applied to two larger languages, \textsc{Tiger} and \textsc{MITScript}, which have $80$ and $60$ rules, respectively. It turns out that these do not introduce fundamental new challenges, although they do impose more stress on \mandate's term-rewriting engine.

The actual generated graph patterns and compiled-mode CFG-generators, as well as several example outputs of the interpreted-mode CFG-generators, are available in the supplemental material.

\begin{wrapfigure}{r}{0.5\textwidth}
\vspace{-3em}
\begin{center}
\begin{lstlisting}[language=Haskell,basicstyle=\tiny\sffamily]
genCfg t@(Node "ForExp" [a, b, c, d]) =
  do (tIn, tOut) <- makeInOut t
     (bIn, bOut) <- genCfg b
     (cIn, cOut) <- genCfg c
     (dIn, dOut) <- genCfg d
     (aIn, aOut) <- genCfg a
     connect tIn bIn; connect dOut tOut
     connect dOut dIn; connect dOut tOut
     connect cOut dOut; connect bOut cIn
     return (inNodes [tIn], outNodes [tOut])
\end{lstlisting}
\vspace{-0.5em}
\caption{Generated Tiger for-loop CFG generator}
\label{fig:tiger-for-loop}
\end{center}
\end{wrapfigure}

\paragraph{Reduction State}

The main difference between \IMP and the larger languages is in the structure of their heap. In \IMP, the reduction state was a simple map of variable names to values. In \textsc{Tiger} and \textsc{MITScript}, the reduction state must allow for stack frames, pointers, and closures. This reduction state is merely \textbf{a particularly-shaped pair of environment and term}, and involves \textbf{no extension} to the mechanics already used in \IMP.  We show an example reduction state in \figref{fig:mitscript-heap}, the starting state of all \textsc{MTIScript} programs, which contains hardcoded mappings of several strings to builtin functions. More details about these reduction states are available in \appref{app:tiger-red-state}.

A limitation of the current \mandate implementation is that it does not support associative-commutative matching for nested maps. Hence, each heap record is implemented as an assoc-list rather than a map, with record lookup implemented as a semantic function. This is the reason why the boolean-tracking abstraction is not implemented for \textsc{Tiger} or \textsc{MITScript}.

\paragraph{Examples}

Our first example is \textsc{MITScript} if-statements, which illustrate how \mandate can generate multiple CFG-generators from the same semantics. When run in compiled-mode with the value-irrelevance abstraction, \mandate generates the code in \figref{fig:tiger-if-statements} for if-statements:

Note how the generator returns multiple exit-nodes for the if-statement. This means that the graph-generator will draw edges from both nodes to whatever comes after the if-statement. When run with the expression-irrelevance abstraction \mandate produces code generating smaller graphs which do not have nodes to represent the evaluation of the condition, code shown in \figref{fig:tiger-if-stmt}.

We next show our most impressive example, showcasing \mandate's ability to cut through syntactic sugar to determine the control-flow behavior of a node type: \textsc{Tiger} for-loops.  \textsc{Tiger} for-loops have more parts than for-loops in most languages. A for-loop in \textsc{Tiger} looks like this.

\begin{figure*}
\centering
\begin{subfigure}[t]{0.5\textwidth}
\centering
\begin{lstlisting}[language=Haskell,basicstyle=\scriptsize\sffamily]                                  
genCfg t@(Node "If" [a, b, c]) =
  do (tIn, tOut) <- makeInOut t
     (aIn, aOut) <- genCfg a
     (bIn, bOut) <- genCfg b
     (cIn, cOut) <- genCfg c
     connect tIn aIn
     connect aOut cIn
     connect aOut bIn
     return (inNodes [tIn],
             outNodes [bOut,cOut])
\end{lstlisting}
\subcaption{}
\label{fig:tiger-if-statements}
\end{subfigure}
~
\begin{subfigure}[t]{0.5\textwidth}
\centering
\begin{lstlisting}[language=Haskell,basicstyle=\scriptsize\sffamily]
genCfg t@(Node "If" [_, a, b]) =
  do (tIn, tOut) <- makeInOut t
     (bIn, bOut) <- genCfg b
     (aIn, aOut) <- genCfg a
     connect tIn bIn
     connect tIn aIn
     return (inNodes [tIn],
             outNodes [bOut,aOut])
\end{lstlisting}
\vspace{1.53\baselineskip}
\subcaption{}
\label{fig:tiger-if-stmt}
\end{subfigure}
\caption{}
\end{figure*}

\begin{center}
\begin{lstlisting}[language=Algol,basicstyle=\scriptsize\sffamily]
for a = b to c do
  d
\end{lstlisting}
\end{center}

\noindent The semantics are given by a single rule which desugars the above to:

\begin{center}
\begin{lstlisting}[language=Algol,basicstyle=\scriptsize\sffamily,morekeywords={let}]
let a := b
    __hi := c
in
  while (a <= __hi) do
    (d; a := a + 1)
\end{lstlisting}
\end{center}

\mandate generates a graph pattern containing $46$ nodes (shown in \appref{app:tiger-graphpat}). \mandate's code-generator projects these into just $8$ states, for the enter and exit nodes of $b$, $c$, $d$, and the entire loop, yielding the code in \figref{fig:tiger-for-loop}.

\noindent \mandate has blown that giant graph down into just 5 edges. While \mandate does not sort the \texttt{connect} statements, and does output one edge twice, it is still easy to see what the code is doing: it states that control first evaluates $b$, then $c$. The \texttt{connect cOut dOut} line is the most interesting one: it says that, after evaluating $c$ (the upper boundary of the loop), control flows to the thing that happens after $d$ is evaluated, namely the condition of the \texttt{while} loop, from which control flows either to the body of $d$ or to the end of the entire loop.

\noindent It is impressive that \mandate can generate short code for this construct with no reference to these internal computations, particularly considering that \texttt{while} is defined by expansion into \texttt{if}. \textbf{This code is generated completely automatically} from the \textsc{Tiger} semantics and the (function-skipping variant of the) value-irrelevance abstraction. The user need not even provide a projection; one is generated automatically by the code generator, by greedily merging nodes as described in \secref{sec:syntax-directed}.

More examples are available in the supplementary material.

\section{Related Work}

\todo{More self-citations}


\paragraph{Work with Related Goals}

Ours joins a small-but-growing body of work on mechanizing the generation of programming tools. Others include generating static analyzers via multiple executions of interpreters or semantics \cite{darais2015galois,sergey2013monadic,darais2017abstracting,DBLP:journals/pacmpl/BodinGJS19}, using an executable language semantics as a tool (by the K Framework \cite{rosu2010overview}), and work in the language-parametric construction of programming tools such as program transformations \cite{koppel2018one,lin2017quixbugs}. Our work is similar to tools built with the K Framework in that both start with a semantics; ours differs by transforming the semantics into an applied tool, whereas K is limited to applications directly based on equational reasoning, namely interpreters and symbolic executors. Bodin et al \cite{DBLP:journals/pacmpl/BodinGJS19} briefly mention having a 0-CFA for the lambda calculus generated from their ``skeletal semantics;" in personal communication, they described it as ``working code, but not a principled approach" and ``We don't do a CFG-generator."





\paragraph{Transformation of Semantics}

There are a few projects that transform semantics for one language into semantics for a related language. Examples include transforming rules to support gradual types \cite{cimini2016gradualizer,cimini2017automatically}, and deriving new rules for types and scoping of syntactic sugar \cite{pombrio2017inferring,pombrio2018inferring}. Supporting these, there are several languages for defining semantics \cite{lakin2007metalanguage,mulligan2014lem,klein2012run}.

\label{sec:really-explains-danvy}
More closely related are projects that transform semantics presented in one formalism into identical semantics in a different formalism, mostly done by Danvy and his students \cite{danvy2008defunctionalized,danvy2004refocusing,DBLP:journals/lisp/XiaoSA01,ager2003functional,danvy2012inter,danvy2010inter,lgorzata2006derivational}, with a few by others: \citet{hannan1992operational,poulsen2014deriving,huizing2010small,vesely2019one}.

At the start of this project, we attempted to build on this prior work. However, we were surprised to find that \textbf{there is no prior published algorithm for converting SOS to abstract machines}, and found the existing reduction-semantics-to-AM algorithm too weak to be used off-the-shelf as an intermediate step. After discovering PAM, we found it easier to work with: proving the ``up-rules invertible'' property of PAM is done automatically by 20 lines of code, whereas it took a 20-page paper to explain how to mechanically (not automatically) prove the corresponding property for reduction semantics \cite{DBLP:journals/lisp/XiaoSA01}. See \appref{app:related-work} for more discussion of attempts to use prior work. To our knowledge, the \textsc{Tiger} and \textsc{MITScript} languages in this paper are by far the largest languages to undergo automated conversion between two forms of semantics; prior work focuses on simple lambda calculi.


\paragraph{Abstracting Abstract Machines (AAM)}

There are several AAM projects \cite{DBLP:conf/sas/MidtgaardJ08,Midtgaard:2009:CAF:1596550.1596592, van2010abstracting,wei2018refunctionalization,johnson2014abstracting,johnson2013optimizing}, with the earliest example arguably being \citet{jones1981flow}. We explain in \appref{app:related-work} how their abstracting techniques differ fundamentally from abstract rewriting, and why this means that a control-flow analysis based on AAM will not resemble a CFG.


\paragraph{Abstract Rewriting}

Abstract rewriting was introduced by Bert and Echahed in the early 90's \cite{bert1993abstract,DBLP:conf/slp/DidierE95} and has received little attention since. We can thus only compare to their work. While the details differ substantially owing to their different focus (approximating the possible normal forms of a term), it has some common elements with our development: they split nodes into ``constructors'' and ``completely-defined operators,'' resembling our value/nonvalue split, and use a $\top$ node with similar semantics to our $\star$. A major point of departure in their development is that, in their system, each abstract step must overapproximate all concrete transitions from an abstract term. A newer related technique from a different lineage is rewriting modulo SMT \cite{rocha2017rewriting}, which operates on (numeric) symbolic terms constrained by an SMT formula (e.g.: linear arithmetic). We are interested in future work combining these techniques to automatically derive more-precise analyses.

\todo{The previous sentence can be deleted for space}


In \secref{sec:syntax-directed}, we explained how abstract rewriting is
similar to narrowing, but different in an important way. There is a
long tradition in narrowing of proving lifting lemmas similar to our
own; the first comes from \citet{hullot1980canonical}.

\todo{I can put a note here that Bert \& Echahed's analysis would be useful for solving some of the warts in my system}

In \appref{app:related-work}, we discuss the use of several terms similar to ``abstract matching'' in the contexts of abstract interpretation of logic programs, 
model-checking, 
and term-rewriting engines, 
as well as several minor uses of syntactic abstraction.

\paragraph{Control-Flow Analysis}

Many papers have been written on control-flow analysis  \cite{Midtgaard:2012:CAF:2187671.2187672,shivers1991control,jones1981flow,DBLP:conf/popl/JagannathanW95}. Older research tries to manually construct a complicated analysis of programs with highly-dynamic control flow. Our work automatically constructs CFG-generators from first principles. Our goal is not to analyze complex programs, but to match the work of hand-written CFG-generators with minimal user input.

As such, we do not consider this work as part of the literature on control-flow analysis. Owing to their different emphasis, these works uniformly have three limitations that make them unsuitable for automatically deriving CFG-generators: 

\begin{enumerate}
\item While they explain how a human could define a new analysis for different languages, their analyses are ultimately manually defined for each language. They further repeat this manual construction for every abstraction used.
\item They check that their result safely approximates executions, but pay no attention to the shape of the graph.
\item Most importantly, they manually partition program states into equivalence classes. That is, they manually annotate the program with labels or program-points, using these as CFG nodes. This is a hindrance to both automation and theory, as most type theories do not contain labels. A major motivation of this work is to support our ongoing work on combining analyses with different notions of program point (i.e.: partition program states differently), which makes not hardcoding them especially important.
\end{enumerate}

\noindent \appref{app:related-work} gives a brief overview of control-flow analysis, and discusses two works that deserve special mention, due to use of abstract machines and focus on generality.

\section{Conclusion}

This work presented both an algorithm for constructing CFGs
from first principles and the world's first CFG-generator
generator. Yet our work also furthers three larger goals.

First, we have provided an answer to ``what is a
control-flow graph?'' beyond the vague ``a CFG is an abstraction of control-flow:'' A CFG is
a projection of the transition graph of abstracted abstract
machine states. This fulfills our original impetus for this work, that
of needing to create static analyzers with exotic notions of ``program point.''

Second, we have introduced abstract rewriting as a simple yet powerful
technique for deriving tools from a language's semantics. We are excited by the idea of using it to
derive other artifacts from language semantics, such as a symbol-table
generator from the typing abstract-machine \cite{sergey2011type}.

Third, we have used a language's semantics to derive a tool
entirely unlike a semantics. Though it's long been known that a semantics can be executed to obtain an interpreter or even a
symbolic-executor \cite{rosu2010overview}, we see our
contribution as qualitatively different, and an important
step towards the dream of being able to write down a language's
syntax and semantics and automatically derive all desired tools.

Mandate is available from \url{https://github.com/jkoppel/mandate}.

\begin{scriptsize}
\textbf{Acknowledgments:}
 This material is based upon work supported by the US Air Force,
AFRL/RIKE and DARPA under Contract No. FA8750-20-C-0208.  Any
opinions, findings and conclusions or recommendations expressed in
this material are those of the author(s) and do not necessarily
reflect the views of the US Air Force, AFRL/RIKE or DARPA.
\end{scriptsize}


\bibliography{koppel_citations.bib}

\clearpage

\appendix

\section{Correctness of SOS-AM Translation}
\label{app:correctness}

This section proves the correspondence between the operational semantics and abstract machine. We begin by proving the correspondence between operational semantics and PAM. We then prove an important consequence of invertible up rules, and then prove the correspondence between PAM and the abstract machine.

The core idea of the correspondence is simple: The PAM emulates the SOS because each PAM rule was explicitly constructed to correspond to an RHS fragment of the SOS. The PAM and AM are equivalent because the AM merely removes some redundant steps from the PAM, and because the fused rules in the AM each correspond to several rules in the PAM. However, a PAM derivation may have some ``false starts'' corresponding to a partially-applied SOS rule, and so there are some additional technical details to explain which states are included in the correspondences.

\subsection{SOS-PAM Correspondence}

The goal of this section is to prove the main theorem of SOS-PAM correspondence:

\thmsospam*

The forward direction is easy, because the PAM rules were designed to follow in lockstep with each component of the SOS rules. The reverse-direction appears harder, but is rendered easy by two important facts:

\begin{observation}
\label{obs:pam-gen-invert}
From the LHS of each PAM rule, it is possible to identify the arguments $s,K,\text{rhs}$ of $\textsc{sosRhsToPam}$ that generated it.
\end{observation}

This is because each case generates rules in a distinct form, and each generated rule contains all of the parameters of $\textsc{sosRhsToPam}$. (Note that LHSs originating on line (1) of \figref{fig:sos-to-pam} must be non-values, while those from line (4) must be a value.)

\begin{property}[Sanity of Phase]
The following three properties hold:
\begin{enumerate}
\item If $\pam{c}{K}{\updown_c} \pamto \pam{c\p}{K\p}{\updown_{c\p}}$ and $K\p$ contains strictly more stack frames than $K$, then $\updown_c=\updown_{c\p}=\down$, and $K\p=K\frcomp f$ for some $f$.
\item If $\pam{c}{K}{\updown_c} \pamto \pam{c\p}{K\p}{\updown_{c\p}}$ and $K\p$ contains strictly fewer stack frames than $K$, then $\updown_{c\p}=\up$, and $K=K\p \frcomp f$ for some $f$.
\item If the PAM rules for language $l$ have no up-down rules, and $\pam{c}{K}{\up} \pamto_l \pam{c\p}{K\p}{\updown_{c\p}}$ without using the reset rule, then $\updown_{c\p}=\up$, and $K=K\p \frcomp f$ for some $f$.
\end{enumerate}
\end{property}

These properties follow by inspection of the possible rules. We now prove Theorem \ref{thm:sos-pam} as a corollary of a stronger result:



\begin{lemma}
\label{lem:pam-sos-strengthened}
$c_1 \sosto_l c_2$ if and only if, for all $K$, there is a derivation $\pam{c_1}{K}{\down} \pamto_l^* \pam{c_2}{K}{\up}$ which does not use the reset rule. This derivation must use the same sequence of rules regardless of $K$.
\end{lemma}
\begin{proof}

We address each direction.

$(\Rightarrow)$: Let $K$ be an arbitrary context, and consider a derivation of $c_1 \sosto c_2$. By induction on the last SOS rule applied, we prove there exists a derivation $\pam{c_1}{K}{\down} \pamto^* \pam{c_2}{K}{\up}$. 

Let the last SOS rule used take the form $c_1 \sosto C_1[\dots C_n[c_2]]$, where each $C_i$ is a single RHS fragment. We define $R$ to be the remaining fragments, $R=C_{i+1}[\dots C_n[c_2]]$. We induct again on $i$ to show that $\pam{c_1}{K}{\down}\pamto^* s_i$, where:

\begin{enumerate}
\item If $i=0$ (base case), then $s_i=\pam{c_1}{K}{\down}$.
\item If $C_i$ is of the form $\letstepto{c}{c\p}{\hole}$, then  $$s_i= \pam{c\p}{K\frcomp \cxtfr{c\p}{R}}{\up}$$
\item If $C_i$ is of the form $\letcomp{c\p}{f(\overline{c})}{\hole}$, then $$s_i = \pam{c\p}{K\frcomp \cxtfr{c\p}{R}}{\down}$$
\end{enumerate}

\noindent and further, the PAM system contains rules generated by $\textsc{sosRhsToPam}(s_i, k, R)$ for some context variable $k$. We handle each case:

\begin{enumerate}
\item Satisfied by the empty transition sequence and definition of \textsc{sosRuleToPam}.
\item Then there is a rule that $s\pamto \pam{c}{K\cxtfr{c\p}{R}}{\down}$. By inversion of the SOS derivation, we must have that $c\sosto c\p$. Then, by the outer induction hypothesis, $\pam{c}{K\cxtfr{c\p}{R}}{\down} \pamto^* \pam{c\p}{K\cxtfr{c\p}{R}}{\up}$. The rest follows by line (3) of the definition of \textsc{sosRhsToPam}.
\item Then there is a rule that $s \pamto \pam{c\p}{K\cxtfr{c\p}{R}}{\down}$. The rest follows by line (4) of the definition of \textsc{sosRhsToPam}. 
\end{enumerate}

The inner induction hypothesis for $C_n$ tells us that a rule $s \pamto \pam{c_2}{K}{\up}$ must exist, finishing the proof of the outer induction.

\vspace{1em}

$(\Leftarrow)$: We proceed by a strong induction on all derivations of the form $\pam{c_1}{K}{\down} \pamto_l^* \pam{c_2}{K}{\up}$. Consider the first PAM rule of the derivation. Because it may not depend on $K$, it must have an LHS generated on line (1) of \figref{fig:sos-to-pam}. By Observation \ref{obs:pam-gen-invert}, we can hence reconstruct the entire originating SOS rule, $c_1 \sosto C_1[\dots C_n[c_2]]$. We show that $c_1 \sosto c_2$ by this rule. 

The proof proceeds similarly to the forward direction, so we omit more details. We induct over $i$, and show that there must be a prefix of the derivation $\pam{c_1}{K}{\down} \pamto^* s_i$, where $s_0=\pam{c_1}{K}{\down}$ and $s_i$ is a state corresponding to the computation of $C_i$. Each $s_i$ matches the LHS of the PAM rule used in the derivation; Observation \ref{obs:pam-gen-invert} tells us this rule must be the one generated for $C_{i+1}$. The only interesting case is for recursive steps; there, $s_i=\pam{c}{K\p}{\down}$, and the Sanity of Phase properties dictate there must be a later state in the derivation $\pam{c\p}{K\p}{\up}$; applying the outer inductive hypothesis finishes this case.

\end{proof}



\subsection{Invertibility}

Our goal is to show that, if all up-rules are invertible (Definition \ref{def:up-rule-invertible}), then $\pam{c}{K}{\up}\pamto^* \pam{c}{K}{\down}$ when $c$ is a non-value, justifying the optimizations of \secref{sec:pam-to-am}, and characterizing which PAM states are and are not removed when translating a derivation to AM. However, there are a few technical restrictions on this.

\begin{definition}
A configuration/context pair $(c,K)$ is \textbf{non-stuck} if $\pam{c}{K}{\up} \pamto^* \pam{c\p}{\halt}{\up}$ for some $c\p$.
\end{definition}

Because each PAM rule corresponds to part of an SOS rule, our definition of non-stuckness is different from the usual one: it is intended to exclude terms which correspond to a partial match on an SOS rule. A single step $c_1 \sosto c_2$ in the SOS corresponds to a sequence $\pam{c_1}{\halt}{\down} \pamto^* \pam{c_2}{\halt}{\up}$ in the PAM, so a state is non-stuck if it can complete the current step. Stuck states result from SOS rules which only partially match a term. For example, the SOS rule

\begin{mathpar}
(\assign{a.b}{v}, \mu) \sosto \letcomp{(r,\mu\p)}{\text{Lookup}((a,\mu))}{ \\\vspace{-0.5em} \ \ \qquad\letcomp{\false}{\text{ContainsField}(r, b)}{(\error, \mu)}}
\end{mathpar}

\vspace{1.5em}

\noindent decomposes into 3 PAM rules. Assuming $\text{Lookup}$ succeeds, the first PAM rule brings $\pam{(\assign{a.b}{v}, \mu)}{K}{\down}$ into the state 

\begin{small}
$$\pam{(r,\mu\p)}{K\frcomp \shortcxtfr{\letcomp{\false}{\text{ContainsField}(\hole_t,b)}{(\error, \mu)}}}{\down}$$ 
\end{small}

\noindent If $\false \neq \text{ContainsField}(r,b)$, then this will be a stuck state.

Excluding stuck states is enough to prove the general Invertibility Lemma:

\begin{lemma}[Invertibility]
\label{lem:invertibility}
If all up-rules for $l$ are invertible, and there are no up-down rules for $l$ other than the reset rule, then, for any non-stuck non-value $(c, K)$, $\pam{c}{K}{\up} \pamto^* \pam{c}{K}{\down}$.
\end{lemma}
\begin{proof}
Consider a derivation $\pam{c}{K}{\up} \pamto^* \pam{c\p}{\halt}{\up}$. Because there are no up-down rules, each step must follow from an up-rule. Hence, each step is invertible. Applying each inverted step gives a new derivation $\pam{c\p}{\halt}{\down} \pamto^* \pam{c}{K}{\down}$.

This requires that the term in the RHS of each step is a non-value, which also follows because each rule is invertible, and hence the RHS can be reduced.
\end{proof}

\noindent This motivates the definition of an \keyterm{inversion sequence}. We add the condition about the reset rule to prevent the definition from including arbitrarily large subsequences of a nonterminating execution.

\begin{definition}
An \keyterm{inversion sequence} which begins at $\pam{c}{K}{\up}$ is a sequence of transitions $\pam{c}{K}{\up} \pamto^* \pam{c}{K}{\down}$ which contains at most one application of the reset rule.
\end{definition}

This idea of an inversion sequence partitions a derivation $\pam{c_1}{K_1}{\down}\pamto^*\pam{c_2}{K_2}{\down}$ into two parts: the inversion sequences, which do redundant work, and the remainder, which we call the \keyterm{working steps}. Sometimes a derivation must be extended to contain a complete inversion sequence, which is then eliminated upon the conversion to an abstract machine.

\begin{definition}
A reduction $\pam{c_1}{K_1}{\updown_1}\pamto \pam{c_2}{K_2}{\updown_2}$ within a derivation is a \textbf{working step} if the derivation cannot be extended so that $\pam{c_1}{K_1}{\updown_1}$ is part of an inversion sequence.
\end{definition}

\begin{observation}
If $(c,K)$ is a non-stuck non-value and all up-rules are invertible, then, by the Determinism Assumption, all sequences $\pam{c}{K}{\up} \pamto^* \pam{c\p}{K\p}{\updown}$ may be extended to contain an inversion sequence starting at $\pam{c}{K}{\up}$.
\end{observation}

\begin{corollary}
\label{corr:no-inversion-value}
If a reduction $\pam{c}{K}{\up} \pamto \pam{c\p}{K\p}{\up}$ cannot be extended to contain an inversion sequence starting at $\pam{c}{K}{\up}$, then either $(c,K)$ is stuck or $c$ is a value.
\end{corollary}

With these extra properties, we are now ready to exactly state the PAM-AM correspondence.

\subsection{PAM-AM Correspondence}

Intuitively, a derivation in the AM $\am{c_1}{K_1}\amto^* \am{c_2}{K_2}$ is the same as a derivation in the PAM, but with the inversion sequences cut out and with some consecutive steps merged into one.  We prove this in steps. First, we show that if $\pam{c_1}{K_1}{\down}\pamto^* \pam{c_2}{K_2}{\down}$, and the last transition is a working step, then  $\am{c_1}{K_1}\unfusedamto^* \am{c_2}{K_2}$. Next we prove that every derivation in the Unfused AM corresponds to a derivation in the fused AM, but with some consecutive steps merged. We also prove the reverse theorem, which is easier to state, as every state in an AM derivation has a corresponding state in the PAM.

The forward direction comes first. This version of the theorem starts with transitions between down-states, to simplify consideration of which states may be eliminated on conversion to AM.

\thmpamunfusedamforward*

\begin{proof}
First, recall that, if $\unfusedamto_l$ is defined, then all up-rules for $l$ are invertible.

Consider a derivation $\pam{c}{K}{\down}\pamto_l^* \pam{c\p}{K\p}{\down}$. Remove all maximal inversion sequences. Then remove all phases from the PAM states, resulting in AM states. This means that an inversion sequence $\pam{c_1}{K_1}{\up}\pamto^* \pam{c_1}{K_1}{\down}$ is replaced with a single state $\am{c_1}{K_1}$.

If we can show that all PAM rules for all remaining steps of the derivation have a corresponding rule in the unfused abstract machines, then we will be done. Note that the only PAM rules without a corresponding rule in the AM are those of the form $\pam{c}{K}{\down} \pamto \pam{c}{K}{\up}$, and those of the form $\pam{c_1}{K_1}{\up} \pamto \pam{c_2}{K_2}{\up}$ where $c_1$ is a non-value. The former correspond to stutter steps in the AM and may be ignored. For the latter, because of the Determinism Assumption and Corollary \ref{corr:no-inversion-value},  all such transitions must be part an inversion sequence, and were hence removed.
\end{proof}

In the next two proofs, we use the notation $\am{c}{K} \overset{F}{\amto} \am{c\p}{K\p}$ to denote that $\am{c}{K}$ steps to $\am{c\p}{K\p}$ by rule F.



\thmpamunfusedambackward*

\begin{proof}
Consider a derivation $\am{c}{K} \unfusedamto_l^* \am{c\p}{K\p}$. For each rule of the unfused abstract machine used in this derivation, consider the corresponding PAM rule that generated it.

Let $\updown_c$ be the phase of the LHS of the first such rule. We will show that there is $\updown_{c\p}$ such that $\pam{c}{K}{\updown_c} \pamto_l^* \pam{c\p}{K\p}{\updown_{c\p}}$, obtained by replacing each AM rule with its corresponding PAM rule and by inserting inversion sequences. We proceed by induction on the derivation.

Consider the last step of the derivation $\am{c\pp}{K\pp} \unfusedamto_l \am{c\p}{K\p}$. Consider the AM rule of this last step, and let G be the PAM rule from which it originated, $\pam{c_G^1}{K_G^1}{\updown_G^1} \overset{G}{\pamto_l} \plug{C_G}{\pam{c_G^2}{K_G^2}{\updown_G^2}}$. If G matches, it would finish the proof. In the base case where there is only one step, taking $\updown_c=\updown_G^1$ and $\updown_{c\p}=\updown_G^2$ suffices to make it match.

If there is more than one step in the derivation, then, by the induction hypothesis, there are phases $\updown_c$ and $\updown_{c\pp}$ such that $\pam{c}{K}{\updown_c}\pamto_l^* \pam{c\pp}{K\pp}{\updown_{c\pp}}$. Consider the second-to-last step of the AM derivation $\am{c\ppp}{K\ppp}\unfusedamto_l \am{c\pp}{K\pp}$ and its generating AM rule, and let the corresponding PAM rule be F. If the RHS of F matches the LHS of G, we would be done. We hence must consider all PAM rules F, G such that the RHS of F and LHS of G match except for the phase, meaning they would erroneously match upon conversion to AM rules. Call such an $(F,G)$ a \emph{confused pair}. We perform case analysis on \figref{fig:sos-to-pam} to find all possible confused pairs, and for each find a derivation $\pam{c\pp}{K\pp}{\updown_{c\pp}} \pamto_l^* \pam{c\p}{K\p}{\updown_{c\p}}$.

\figref{fig:sos-to-pam} gives 3 possible forms of PAM RHSs, generated on lines (2), (3), and (4), and 3 possible LHSs, generated on lines (1), (3), and (4). Note that some of these only match values/non-values, and that, because of the prohibition on up-down rules, all rules using the LHS from line (3) will be restricted to only match values. This leaves only 3 possible forms for a confused pair: using the RHS/LHS generated on lines (2)/(1), (2)/(4), and (4)/(3). We analyze each in turn. We find that the first case is desirable, as it results from removing inversion sequences, while the other two are benign, as another rule must exist that does match.

\begin{itemize}
\item (2)/(1): In this case, the RHS of a rule F, $s \overset{F}{\pamto_l} \pam{c_F}{K_F}{\up}$, matches the LHS of a rule G, $\pam{c_G}{K_G}{\down} \overset{G}{\pamto_l} t$, where $c_F=c_G$ upon matching with $c\pp$. $c_F=c_G$ must be a non-value because $c_G$ originates from a SOS rule $c_G \sosto r$, and by the Sanity of Values assumption. The invertibility lemma finishes this case.

\item (2)/(4): In this case, the RHS of a rule F, $s \overset{F}{\pamto_l} \pam{c_F}{K_F}{\up}$, matches the LHS of a rule G, $\pam{c_G}{K_G}{\down} \overset{G}{\pamto_l} t$. Further, after unifying with the current state $\am{c\pp}{K\pp}$,  $\am{c_F}{k_F} = \am{c_G}{k_G}$, and $k_F=k_G$ can be written $k_F = k_G = k \frcomp \cxtfr{c_F}{\text{rhs}}$. By the induction hypothesis, there is a derivation $\pam{c}{K}{\updown_c}\pamto_l^* \pam{c_F}{K_F}{\up}$. Extend the last transition of this derivation to a maximal sequence $\pam{c_F\pp}{K_F\pp}{\down} \overset{H}{\pamto_l} \pam{c_F\p}{K_F}{\down} \pamto_l* \pam{c_F}{K_F}{\up}$ which does not use the reset rule; by the Sanity of Phase properties, this must exist, and Rule H must have been created by line (3). By Observation \ref{obs:pam-gen-invert}, we know that rule G was created by an invocation matching $\textsc{sosRhsToPam}(\pam{c_F}{K_F}{\down}, k, \text{rhs})$, while rule H was created by an invocation matching

\vspace{-1em}
$$\qquad \textsc{sosRhsToPam}(\pam{c_F\pp}{k}{\down}, k, \letstepto{x}{c_F}{rhs})$$

\noindent which means there is also a rule J created by an invocation

\vspace{-1em}
$$\qquad\textsc{sosRhsToPam}(\pam{c_F}{K_F}{\up}, k, \text{rhs})$$

\noindent whose RHS is $t$. Using rule J completes this case.

\item (4)/(3): In this case, the RHS of a rule F, $s \overset{F}{\pamto_l} \pam{c_F}{k_F}{\down}$,  matches the LHS of a rule G, $\pam{c_G}{k_G}{\up} \overset{G}{\pamto_l} \plug{C_G}{t}$. Further, after unifying with the current state $\am{c\pp}{k\pp}$,  $\am{c_F}{k_F} = \am{c_G}{k_G}$, and $k_F=k_G$ can be written $k_F = k_G = k \frcomp \cxtfr{x}{\text{rhs}}$. By Observation \ref{obs:pam-gen-invert}, we know that rule F was created by an invocation

\vspace{-1em}
$$\textsc{sosRhsToPam}(s, k, \letstepto{c_F}{x}{\text{rhs}})$$

\noindent and there is hence a rule H created by an invocation $\textsc{sosRhsToPam}(\pam{c_F}{k_F}{\down}, k, \text{rhs})$. As rule G was created by an invocation $\textsc{sosRhsToPam}(\pam{c_G}{k_G}{\up}, k, \text{rhs})$, rule H hence takes the form $\pam{c_F}{k_F}{\down} \overset{H}{\pamto_l} \plug{C_G}{t}$. Using rule H completes this case.

\end{itemize}
\end{proof}

\noindent Finally, to show the correspondence between the Unfused AM and the normal AM, we must show that fusing rules does not substantially alter the transition relation. This is very simple, thanks to the Fusion Property.

\begin{lemma}
\label{lem:am-fusion-lemma}
Let $M$ be an abstract machine whose transition relation is $\amto_M$, containing a rule F. Let $M\p$ be $M$ with Rule F fused with all possible successors, and let its transition relation be $\amto_{M\p}$. Then, for any state $\am{c}{K}$, $\am{c}{K}\amto_{M\p} \am{c\p}{K\p}$ if and only if $\am{c}{K}\amto_M \am{c\p}{K\p}$, or $\am{c}{K} \overset{F}{\amto_M} \am{c\pp}{K\pp} \overset{G}{\amto_M} \am{c\p}{K\p}$ for some rule $F \neq G$.
\end{lemma}
\begin{proof}
By Property \ref{prop:fusion}.
\end{proof}

\thmunfusedamam*

\begin{proof}
Corollary of Lemma \ref{lem:am-fusion-lemma}.
\end{proof}

Finally, we state some useful properties which are analogues of Sanity of Phase.

\begin{property}[Sanity of Frame]

The following properties hold:

\begin{enumerate}
\item If $\am{c}{K} \amto \am{c\p}{K\p}$, then either $K=K\p \frcomp f$ for some $f$, $K\p = K\frcomp f$ for some $f$, or there are $f$, $f\p$, $K\pp$ such that $K=K\pp \frcomp f$ and $K\p=K\pp \frcomp f\p$.

\item If $c$ is a nonvalue, and $\am{c}{K} \amto \am{c\p}{K\p}$, then $K$ is contained in $K\p$.
\end{enumerate}

\end{property}

\paragraph{Why Didn't We Use Bisimulations?}

A common alternative way of stating results like Theorems \ref{thm:pam-am} and \ref{thm:am-pam} is by giving a stuttering bisimulation between the PAM and the AM. However, between any two (terminating) transition systems, without giving additional labels on the states, there always exists a trivial stuttering bisimulation. The interesting information lies in giving a specific stuttering bisimulation. We decided that dealing with the additional machinery of stuttering bisimulations would add to the background needed to understand the proof, without much benefit.

\section{Proofs of Abstract Rewriting Theorems}

\label{app:abs-rewriting-proofs}

\thmgenabsrewriting*

\begin{proof}
We show a correspondence between the applications abstract rewriting algorithms for both $\beta_1$ and $\beta_2$. We know that the LHS of rule $F$ matches $\am{c_1}{K_1}$ with witness $\sigma$; hence, by the Abstract Matching Property, it also matches $\am{\abstr{c_1}}{\abstr{K_1}}$ with some witness $\sigma\p \succ \sigma$. Let the RHS of $f$ be $\text{rhs}_p$. Then:
\begin{itemize}
    \item If $\text{rhs}_p=\letcomp{c_\text{ret}}{\text{func}(\overline{c_\text{args}})}{\text{rhs}_p\p}$, then the concrete rewriting of  $\am{c_1}{K_1}$ must have picked some $r \in \beta_1(\text{func})(\sigma(\overline{c_\text{args}}))$, where $r$ matches $c_\text{ret}$ with witness $\sigma_r$. Since $\beta_1$ is a base abstraction and $\abstr{\sigma}(\overline{c_\text{args}}) \succ \sigma(\overline{c_\text{args}})$, there is an $\abstr{r}\in \beta_2(\text{func})(\abstr{\sigma}(\overline{c_\text{args}}))$ with $\abstr{r}\succ r$. Then, by the abstract matching property, $\abstr{r}$ matches $c_\text{ret}$ with witness $\abstr{\sigma_r}\succ \sigma_r$. The abstract rewriting algorithm then proceeds to recursively evaluate $\text{rhs}_p$ with some $\abstr{\sigma\p}$ and $\sigma\p$ respectively; by their definition and the previous argument, we must have $\abstr{\sigma\p}\succ \sigma\p$.
    
    \item If $\text{rhs}_p=\am{c_p\p}{K_p\p}$, then the result is  $\am{\abstr{\sigma}(c_p\p)}{\abstr{\sigma}(K_p\p)}$. Since $\abstr{\sigma}\succ \sigma$, $\am{\abstr{\sigma}(c_p\p)}{\abstr{\sigma}(K_p\p)} \succ \am{\sigma(c_p\p)}{\sigma(K_p\p)} = \am{c_2}{K_2}$ by the Abstract Matching Property, finishing the proof. 
\end{itemize}
\end{proof}

\thmabstransition*

\noindent The proof uses the following observation:

\begin{observation}
\label{obs:noncut-prop}
The nontermination-cutting relation $(\triangleleft)$ satisfies the following properties:

\begin{enumerate}
    \item If $a \prec b$ and $b \triangleleft c$, then $a \triangleleft c$.
    \item If $a \rightarrow b$ and $a \triangleleft c$, then $b \triangleleft c$.
\end{enumerate}
\end{observation}

\noindent We now prove the Abstract Transition Theorem:

\begin{proof}

Because $(\prec)$ is reflexive and transitive, there must be a \textit{canonical derivation} of $a \sqsubseteq \alpha(a)$ of one of the following three forms:

\begin{itemize}
    \item Case (1): $a \prec \alpha(a)$. Then, by the lifting lemma, there is a $g$ with  $b \prec g$ and $\alpha(a) \absred{\beta} g$.
    \item Case (2): There are $x_1, x_2$ such that $a \prec x \triangleleft x_2 \sqsubseteq \alpha(a)$. Then, by Observation \ref{obs:noncut-prop}, $a \triangleleft x_2$, and hence $b \triangleleft x_2$, and hence $b \sqsubseteq x_2 \sqsubseteq \alpha(a)$.
    \item Case (3): There is an $x_1$ such that $a \prec x_1$, and, for all $x_2$ such that $x_1 \absred{\beta} x_2$, $x_2 \sqsubseteq \alpha(a)$. First, by the lifting lemma, there is an $x_1\p$ with $b \prec x_1\p$ and $x_1 \absred{\beta} x_1\p$.  But, by assumption, $x_1\p \sqsubseteq \alpha(a)$. Hence, $b \prec x_1\p \sqsubseteq \alpha(a)$, so $b \sqsubseteq \alpha(a)$.
\end{itemize}

\end{proof}

\section{More Details on Syntax-Directed CFG-Generators}
\label{app:notes-syntax-directed}

\subsection{Correctness of Graph Patterns}
\label{app:graph-pattern-correctness}

We first note that, by starting abstract execution on a node whose
immediate children are all non-value variables, this algorithm assumes
that the initial state of the program must contain no value nodes. This is not a real restriction; one can
transform any language to meet this criterion by replacing each
value-node $V(\overline{x})$ in initial program states with a
non-value node $MkV(\overline{x})$ and adding the rule
$MkV(\overline{x}) \sosto V(\overline{x})$. So, pedantically speaking,
the graph-patterns produced by the algorithm are not actually graph patterns of the
original language, but rather of this normalized form.

We now build the setting of our proofs. In this develpment, let
$\set{term}_\text{Var}$, $\set{amState}_\text{Var}$, etc be variants
of $\set{term}\star$, $\set{amState}\star$, etc that may contain free variables as well as $\star$ nodes. We extend the $\prec$
ordering to include the subsumption ordering, and with the relation $x_\text{mt} \prec \star_\text{mt}$ for any variable $x$, so that $a\prec b$ if
$a$ may be obtained from $b$ by specializing a match type, expanding a
star node, \textit{or} substituting a variable. 

As mentioned in \secref{sec:abstractions}, we add a few technical conditions to the definition of an abstraction $\alpha$. The first condition prevents an antagonistically-chosen $\alpha$ from doing something substantially different when encountering a more abstract term.

\begin{assumption}
\label{assm:alpha-monotone-prec}
On terms without variables, $\alpha$ must be monotone in the $\prec$ ordering. As the analogue for terms with variables, if there is a substitution $\sigma$ such that $b=\sigma(a)$, then there must be a substitution $\sigma\p$ extending $\sigma$ such that $\alpha(b)=\sigma\p(\alpha(a))$.
\end{assumption}

The next condition is stronger than we need, but greatly simplifies the discovery of the correspondence between graph patterns and interpreted-mode graphs. In plain words, it states that abstractions may not drop stack frames: they may skip over the execution of a subterm entirely, but may not skip over only the latter part of a computation. All abstractions discussed in this paper satisfy it.

\begin{definition}
Define the \textit{stack length} of a state $s=\am{c}{K}$ as:
\begin{itemize}
    \item $\text{stacklen}(\am{c}{\halt}) = 0$
    \item $\text{stacklen}(\am{c}{K \frcomp f}) = 1 + \text{stacklen}(\am{c}{K})$
\end{itemize}
\end{definition}

\begin{assumption}
\label{assm:no-frame-drop}
For all $a\in \set{amState}_\text{Var}$, we require that $$\text{stacklen}(a) \le \text{stacklen}(\alpha(a))$$

\noindent Further, there must be a derivation of $a \sqsubseteq \alpha(a)$ where none of the intervening states $c$, as in Definition \ref{defn:abstraction-ordering}, satisfy $\text{stacklen}(c) < \text{stacklen}(a)$.
\end{assumption}


We now begin the proofs. We need a new version of the Generalized Lifting Lemma for narrowing. Mimicking the proof, and using the relation between matching and unification, gives the following:

\begin{lemma}[Lifting Lemma (Narrowing)]
Let $a\in\set{amState}_\star, b \in \set{amState}_\text{Var}$, $a \prec
b$, and let $\beta$ be a base abstraction. Suppose $a\absred{\beta} a\p$. Then there exists $b\p$ such that $b\absnarrow{\beta} b\p$ and $a\p\prec b\p$.
\end{lemma}

We now prove the correspondence between a graph pattern and the relevant subgraph of an abstract transition graph. To isolate the relevant subgraphs, we use the concept of \textit{hammocks} from graph theory, which are commonly used in the analysis of control-flow graphs (e.g.: \citet{Ferrante1987}). A hammock of a control-flow graph is a single-entry single-exit subgraph. We use the modified term \textit{weak hammock} to refer to a single-entry multiple-exit subgraph.

\begin{definition}
Let $N^+(n)$ be the out-neighborhood of $n$ in a graph $G$. Then the \textit{weak hammock} of $G$ bounded by entry node $n$ and exit node-set $\mathcal{T}$ is the subgraph of $G$ induced by the node set given by the least-fixed-point of $Q$, where

$$Q(S) = \{n\} \cup \bigcup_{m\in (S \setminus \mathcal{T})} N^+(m) $$
\end{definition}

Our goal now is to, given the abstract transition graph of a program, discover the fragment that corresponds to the control-flow of a single node. We will then prove the correspondence between these fragments and the relevant graph pattern.

\begin{definition}
Let $N$ be a non-value node type and $\alpha$ an abstraction, and consider some configuration $S=\am{(N(\overline{e_i}),\mu)}{K}$. Let $T$ be  the abstract transition graph $T$ of $(\absred{\alpha})$ starting from $S$. Let $E = \{t \in T | t \neq S \wedge \text{stacklen}(t) \le \text{stacklen}(S) \}$. Then the \textit{CFG fragment for } $S$ is defined inductively as follows:

\begin{itemize}
  \item If none of the $e_i$ are non-values, then the CFG fragment for $S$ is the weak hammock of $T$ bounded by $S$ and $E$.
  \item Otherwise, the CFG fragment for $S$ is the weak hammock of $T$ bounded by $S$ and $E$, minus the edges of the CFG fragments for each non-value $e_i$, minus also the nodes which then become unreachable from both $S$ and $E$. 
\end{itemize}

We say there is a \emph{transitive edge} from the start state of each sub-CFG-fragment to its end states. By default when discussing the edges of a CFG fragment, we do not include the transitive edges.
\end{definition}

\begin{lemma}
Let $\abstr{e} \prec \nonvalstar$. Then, for any $\abstr{s}, \abstr{K}$, $\am{(\abstr{e},\abstr{s})}{\abstr{K}} \triangleleft \am{(\valstar, \top_l)}{\abstr{K}}$.
\end{lemma}
\begin{proof}
Consider $\am{c}{K} \in \gamma(\am{(\abstr{e},\abstr{s})}{\abstr{K}})$. By the Sanity of Frame properties, for any derivation $\am{c}{K}\amto^* \am{c\p}{K\p}$, either $K\p$ contains $K$ or there is a subderivation of the form $\am{c}{K} \amto^* \am{c\p}{K}$.
\end{proof}

\begin{theorem}[Correctness of Graph Patterns]
\label{thm:graph-pat-correctness}
Let $N$ be a non-value node type, and $P$ be its graph pattern under $(\absred{\alpha})$. For an abstraction $\alpha$, consider the abstract transition graph $T$ of $(\absred{\alpha})$ from some start state $S=\am{(N(\overline{e_i}),\mu)}{K}$. Let $F$ be the CFG fragment for $S$ in $T$. Let $\sigma$ be the substitution resulting from unifying the start state of $P$ with $S$. Then, for every edge
$a \absred{\alpha} b$ in $F$, there are $a\p, b\p \in P$ with
$a\prec \sigma(a\p)$, $b\prec \sigma(b\p)$ such that $a\p \absred{\alpha}
b\p$ is in $P$.
\end{theorem}
\begin{proof}
For any $a\in V(F), a\p \in V(P)$ with $a\prec \sigma(a\p)$, this is true for all
edges reachable from $a$ by the Lifting Lemma and by Assumption \ref{assm:alpha-monotone-prec}. We hence must show that every node in $F$ is reachable from some node $x$ satisfying $\exists x\p\in P, x\prec \sigma(x\p)$. But note that every node in $F$ is reachable from either $S$ or by the exit nodes of the CFG fragment for one of the $e_i$. By construction, $S$ is
equal to the start state of $P\p$. It remains to show this condition for the exit nodes of the CFG fragments for the $e_i$, i.e.: the nodes which are the target of a transitive edge.

Pick such a node $x=\am{c_x}{K_x}$, and note that, by construction and by Sanity of Frame, $c_x$ must be a value. Then, using the Sanity of Frame properties and Assumption \ref{assm:no-frame-drop}, the source of the corresponding transitive edge(s) must have the form $\am{(e_i,\mu)}{K_x}$. By induction, we must have $e_i\p, \mu\p, K_x\p$ with $e_i\prec\sigma(e_i\p)$ and $K_x\prec\sigma(K_x\p)$ with $\am{(e_i\p, \mu\p)}{K_x\p} \in P$. This node has a transitive edge to $\am{(\valstar, \top)}{K_x\p}$ in $P$, which satisfies $\am{c_x}{K_x} \prec \sigma(\am{(\valstar, \top)}{K_x\p})$, finishing the proof.
\end{proof}

One advantage of the compiled-mode is that it's done once per language/abstraction pair, so any corner case that causes an infinite-loop would be exposed up front. Yet interpreted-mode CFG-generation is always more precise, albeit slower and less predictable. Can we at least guarantee that interpreted-mode generation will terminate, i.e.: result in a finite graph?

It's easy to construct examples where interpreted-mode CFG-generation does not terminate (e.g.: the identity abstraction, and any non-terminating program), so there's no universal termination theorem. But here's a cool result: if compiled-mode generation terminates, then so does interpreted-mode. Intuitively, this is so because graph-pattern generation does the kind of case analysis that's needed to prove termination of interpreted-mode CFG-generation for a single language/abstraction pair

\begin{theorem}[Finiteness]
Consider a node type $N$, its graph pattern $P$, any state of the form   $S=\am{(N(\overline{e_i}),\mu)}{K}$, and the CFG fragment $F$ for $S$ in the abstract transition graph of $S$. If $P$ is finite, then so is $F$.
\end{theorem}

\begin{proof}
Assume $P$ is finite, and suppose $F$ is infinite. By K\"{o}nig's Lemma, either $F$ contains a node of infinite outdegree, or an infinite path. However, because all semantic functions $\text{func}$ have the property that $\text{func}(\overline{c})$ is finite for all $\overline{c}$, the only possibility is for $F$ to contain an infinite path.

Theorem \ref{thm:graph-pat-correctness} establishes a relation $M$ between $F$ and $P$. Since $P$ is finite, this implies that there is a path $f_1 \rightarrow f_2 \rightarrow \dots \rightarrow f_n$ in $F$ where each $f_i$ is distinct, where the corresponding path in $P$ $p_1\rightarrow \dots \rightarrow p_{n-1} \rightarrow p_1$  is a cycle, where each $\rightarrow$ is either a transitive edge, $(\absred{\alpha})$  (in $F$), or $(\absnarrow{\alpha})$ (in $P$).

However, by Theorem \ref{thm:graph-pat-correctness}, there is $\sigma$ such that $\sigma(p_1) = f_1$, and there is $\sigma\p$ extending $\sigma$ such that $\sigma\p(p_1) = f_n$. Since $\sigma(p_1)$ must be ground, that means $\sigma(p_1)=\sigma\p(p_1)$, and hence $f_1=f_n$, contradicting the assumption. Hence, $F$ must be finite.
\end{proof}

\todo{I think there are some technical details involving the relation between star and vars that I'm squishy about, so the ``can extend sigma to sigma' such that..." thing is not entirely justified. But I'm pretty confident in this proof. I declare myself done at last.}

As a corollary, we obtain our automated proof of termination for the interpreted-mode CFG generators.

\thmbothfinite*

\subsection{Code-Generation}
\label{app:code-gen}

The code generator traverses the graph pattern, identifying subterm
enter and exit states, and greedily merges unrecognized intermediate
states with adjacent ones if one dominates the other, essentially building a custom projection. The output projection will have the property that each equivalence class of nodes forms a connected loop-free subgraph of the entire graph pattern. In the end, for an AST-node with $k$ children, there will be up to $2k+2$ graph nodes, an enter and exit CFG-node for the outer AST node and each child. Once all abstract states have been merged into these graph nodes, the algorithm identifies the edges between the nodes, and outputs code like in \figref{fig:front-page-result}.

The code generator guarantees that the generated code will generate a valid CFG (i.e.: a valid projection of the abstract transition graph). However, there are cases where a projection of the desired form does not exist, and hence the search terminates with failure. This occurs in cases where
a single step in the source program corresponds to an arbitrary number of steps
in the internal semantics (as in Java method resolution,
which must traverse the class table). This scenario does not occur in \textsc{Tiger} and \textsc{MITScript}; it successfully generates code for both.

One subtle fact is that not every AST node will have a distinct exit node. For example, in all three languages, the graph pattern for
$\ite{e}{s_1}{s_2}$ will actually terminate in the exit nodes of $s_1$
and $s_2$, which both must be directly connected to whatever statement
follows the if. The generated CFG-generators hence actually treat
lists of enter/exit nodes as the atomic unit, rather than single nodes. Note that, if the language designer did want to have a distinct ``join'' node terminating the CFG for an if-statement, they could accomplish this by changing the semantics of if-statements to introduce an extra step after the body is evaluated.

\section{Addition Notes on Languages}

\subsection{Notes on Designing the Semantics}

Consider the following rule for evaluating the l-value of a field assignment:

$$(\assign{a.n}{e},\mu) \sosto \letstepto{(a,\mu)}{(a\p,\mu\p)}{(\assign{(a\p).n}{e},\mu\p)} $$

\noindent Upon running \mandate in compiled-mode, rules like this produce graph-patterns which are too coarse. That's because \mandate generates graph patterns for all assignments $\assign{x}{e}$, which could match this rule as well as the rules for assignments of variables or array indices. While we could modify \mandate to produce separate graph patterns for $\assign{a.n}{e}$ vs. $\assign{a[i]}{e}$ or $\assign{x}{e}$, we opted instead to change the semantics. We instead modified the rule to this:

$$(\assign{l}{e},\mu) \sosto \letstepto{(l,\mu)}{l\p,\mu\p)}{(\assign{l\p}{e},\mu\p)}$$

\noindent We then ensured there were separate node types for l-values vs. expressions of the form $a.n$, and defined rules to evaluate each kind of l-value. 

The overall lesson is that \mandate's compiled mode works best when rules match only on the topmost node. We have not yet needed to make \mandate more sophisticated to overcome this problem because it is easier to modify the language semantics.

\subsection{Reduction State for \textsc{Tiger} and \textsc{MITScript}}
\label{app:tiger-red-state}

The main difference between \IMP and the larger languages is in the structure of their heap. In \IMP, the reduction state was a simple map of variable names to values. In \textsc{Tiger} and \textsc{MITScript}, the reduction state must allow for stack frames, pointers, and closures. This reduction state is merely \textbf{a particularly-shaped pair of environment and term}, and involves \textbf{no extension} to the mechanics already used in \IMP. Both \textsc{Tiger} and \textsc{MITScript} use the same design for their reduction state, designed as follows:

\begin{small}
\begin{center}
\begin{mathpar}
\begin{tabular}{lcl}
\set{State} & \gramdef &  $(\text{Stack},\text{Heap})$ \\
\set{HeapAddr} & \gramdef & \set{Int} \\
\set{Stack} & \gramdef & $\text{ConsFrame}(\set{HeapAddr}, \set{Stack}) \mathrel{\sor} \text{NilFrame}$ \\
\set{Heap} & \gramdef & $\set{HeapAddr} \partialfunc \set{Record}$ \\
\set{Record} & \gramdef & $(\set{Symbol} \partialfunc \set{Value}, \set{HeapAddr}?)$
\end{tabular}
\end{mathpar}
\end{center}
\end{small}

\noindent The state consists of a stack and a heap. The stack is a list of heap addresses. The heap is a map of heap addresses to records, where each record contains a map of symbols to values. These records are used to store not only stack frames but also arrays and objects in the language. Each record optionally also contains a pointer to a parent record: looking up a symbol in a record will traverse the parent record if not found.

For an example heap, \figref{fig:mitscript-heap} gives the starting heap of all \textsc{MITScript} programs, which contains hardcoded mappings of several strings to builtin functions. These functions are ordinary \textsc{MITScript} values defined elsewhere, whose bodies contain the special \texttt{Builtin} node, whose evaluation invokes an appropriate semantic function.


A limitation of the current \mandate implementation is that it does not support associative-commutative matching for nested maps. Hence, each heap record is implemented as an assoc-list rather than a map, with record lookup implemented as a semantic function. This is the reason why the boolean-tracking abstraction is not implemented for \textsc{Tiger} or \textsc{MITScript}.

\subsection{Graph Pattern for Tiger For-Loops}
\label{app:tiger-graphpat}

See \figref{fig:forexp-graphpat}. The PDF version of this paper contains the image as a vector graphic so that it can be viewed in full detail upon zoom. It is also available with the supplementary material.

\begin{figure*}
\begin{center}
\vspace{-1.5em}
\includegraphics[scale=0.075]{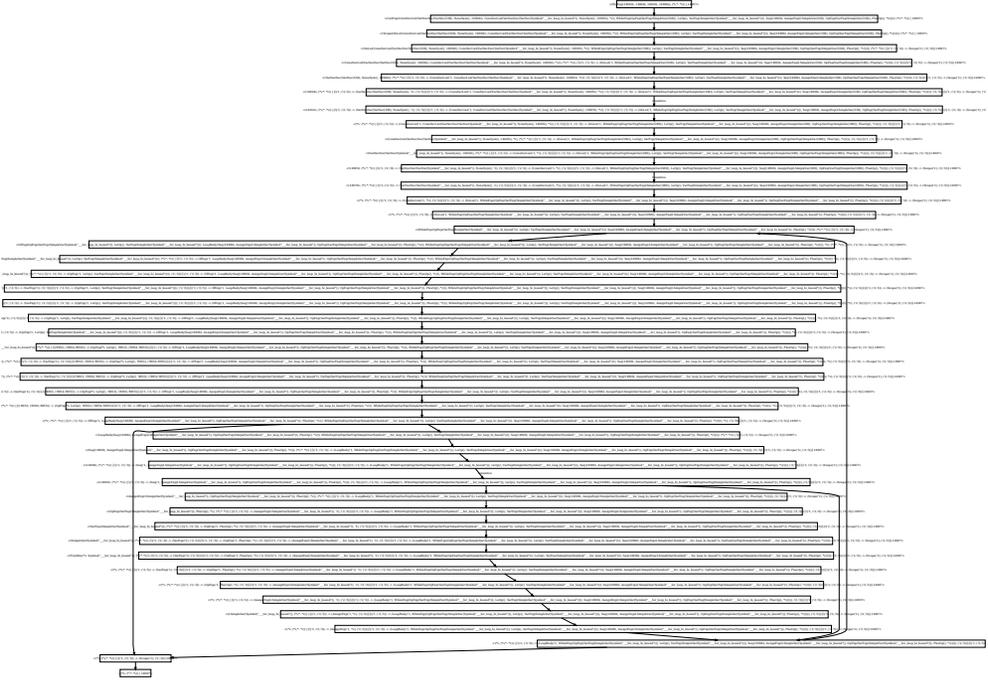}
\end{center}
\caption{Graph pattern for \textsc{Tiger} for-loops}
\label{fig:forexp-graphpat}
\end{figure*}

\begin{center}
\begin{lstlisting}[language=Haskell,basicstyle=\scriptsize\sffamily]
genCfg t@(Node "ForExp" [a, b, c, d]) =
  do (tIn, tOut) <- makeInOut t
     (bIn, bOut) <- genCfg b
     (cIn, cOut) <- genCfg c
     (dIn, dOut) <- genCfg d
     (aIn, aOut) <- genCfg a
     connect tIn bIn
     connect dOut tOut
     connect dOut dIn
     connect dOut tOut
     connect cOut dOut
     connect bOut cIn
     return (inNodes [tIn], outNodes [tOut])
\end{lstlisting}
\end{center}

\section{Additional Discussion of Related Work}
\label{app:related-work}

\paragraph{Transformation of Semantics}

When we first began this project, we planned to simply look up an
algorithm to convert SOS to abstract machines, but surprisingly found
none existed. Almost all work in this space is by Danvy and associates, and, while their papers focus on individual formalisms, Danvy personally told us that he is not interested in small-step semantics because it would not require new techniques over his prior work, and that he is similarly uninterested in creating a general algorithm, when he's already sketched how to do the transformation for a single example \cite{danvy2008defunctionalized}. So, although Danvy may personally know how to do the transformation, \textbf{there is no prior published algorithm for converting SOS to abstract machines}.

Our next plan was to convert the SOS to reduction semantics and then apply Danvy's algorithm for converting reduction semantics to an abstract machine \cite{danvy2004refocusing}. However, we found the assumptions of this algorithm too limiting, as it prohibits any use of external semantic functions or state. Rather than extend this algorithm, we found it much easier to use PAM as an intermediate step, because its nature as a modified term-rewriting system gives us access to unification-based techniques for analyzing and transforming it. For instance, while proving up-rules invertible is a reachability search taking 20 lines of code, proving the equivalent property for reduction semantics, unique decomposition, took a 20-page paper \cite{DBLP:journals/lisp/XiaoSA01}.

The older work by \citet{hannan1992operational} is very limited. They give a few examples of manually deriving an abstract machine from big-step semantics, using a sequence of up to 6 hand-proven transformations. \citet{ager2004natural} continues this work, providing a simpler and automated approach, though with some additional limitations, such as needing nondeterminism to execute if-statements.

\paragraph{Abstracting Abstract Machines (AAM)}

 The first known work explicitly on abstract interpretation of abstract machines was by Midtgaard and Jensen \cite{DBLP:conf/sas/MidtgaardJ08,Midtgaard:2009:CAF:1596550.1596592}, though a modern reader may also consider \citet{jones1981flow} to be an earlier example. This was followed by the ``abstracting abstract machines'' (AAM) line of work began by Van Horn and Might \cite{van2010abstracting,wei2018refunctionalization,johnson2014abstracting,johnson2013optimizing}.

These take a substantially different flavor from our own, due to their focus on higher-order analyses of functional languages. A key technical difference is the choice of abstraction operator: they use variants of store-bounding, whereas we use syntactic abstractions designed to make the abstract state space resemble a conventional CFG. The use of these syntactic abstractions allows our algorithm to automatically abstract the transition rules of an abstract machine via abstract rewriting. Their paper's approach only automatically abstracts reads and writes to an abstract store; the abstract transition steps are still manually defined. For example, the algorithm in our paper can take in a rule like $\am{\true}{k\frcomp\shortcxtfr{(\ite{\hole_t}{s_1}{s_2}, \hole_k)}} \amto \am{s_1}{k}$ and the corresponding rule for $\false$, and deduce that if-statements may step into either branch because both rules match a single $\star$ state. The approach in their paper can at best nondeterministically evaluate an expression to separate $\{\true, \false\}$ states and then match both if-rules, which produces an un-CFG-like graph where the branching happens prior to the if-statement.

\paragraph{Abstract Rewriting}

Outside of Bert \& Echahed, there are a few works which share minor details or terminology with our development of abstract rewriting. In the broadest sense, abstract unification means taking some algorithm that uses unification, and replacing the unification with some other operation. Prior work in this sense comes from work on the abstract interpretation of logic programs \cite{king1995abstract,cousot1992abstract}. These algorithms commonly replace unification with a simple operation such as tracking sharing and linearity, for the purpose of, e.g.:  eliminating backtracking or the occurs-check (see \S10 of \citet{cousot1992abstract} for a literature review).

In contrast, the abstract matching procedure in our work (and in Bert
\& Echahed) is an extension of normal matching to a set of abstract
terms, consisting of normal terms augmented with a set of
``blown-down'' terms. This technique appears absent from the literature
on static analysis of logic programs, for it is not useful for
traditional static analysis by abstract interpretation, as the domain
of abstract terms is infinite (meaning: our algorithm cannot compute a
single control-flow graph which usefully describes all possible
programs).

Apart from Bert \& Echahed, we have found a couple papers that use syntactic abstraction, albeit in a different form. An early use is the ``star abstraction'' of  \citet{DBLP:conf/popl/CodishDG93,DBLP:conf/iclp/CodishFM91}, which merges identical subterms of a tree, and is unrelated to the similarly-named abstraction in this paper. \citet{schmidt1996abstract} uses this to merge identical processes in a CCS-like system, bounding executions to aid in model-checking. Although it is not discussed in the paper, \citet{johnson2013optimizing} does some syntactic abstraction in the implementation, representing all numbers as identical \set{number} nodes.

The term ``abstract matching'' also has an unrelated meaning in the model-checking community, where it refers to finding equivalences between abstract states \cite{puasuareanu2005concrete}.

The ARM abstract rewriting machine \cite{kamperman1993arm} provides a compact instruction set for executing term-rewriting systems efficiently. It handles ordinary rewriting, rather than abstract rewriting in our sense.

\paragraph{Control-Flow Analysis}

Research on control-flow analysis typically focuses on functional languages. Perhaps the most famous work is the k-CFA analysis of \citet{shivers1991control}. This and many other works frame their analysis as an abstract interpretation of executions, though there are too many approaches to describe here; \citet{Midtgaard:2012:CAF:2187671.2187672} gives an extensive survey. 

There are two works from this field which deserve special mention. The first is \citet{jones1981flow}, which is the first to define control-flow as an abstract interpretation of executions. Through a modern lens, it is also the first to do so by abstracting abstract machines: it presents an abstraction of a custom abstract machine resembling a CC machine \cite{felleisen2009semantics}. It shares the limitations mentioned above, although its representation of program points is subtle: it uses a set of tokens representing the different function applications of the source program.

The second is \citet{DBLP:conf/popl/JagannathanW95}, due to its focus on generality, explicit construction of graphs, and attempt to relate their construction to an operational semantics. It is also the most extreme in its use of manual program-point annotations, going as far as to design an abstract machine with an explicit program counter.

\end{document}